\documentclass[journal,twoside, web]{ieeecolor}
\usepackage{generic}
\usepackage{cite}
\usepackage{amsmath,amssymb,amsfonts}
\usepackage{algpseudocode} 
\usepackage{algorithm} 
\usepackage{algorithmicx}
\usepackage{graphicx}
\usepackage{textcomp}
\usepackage{subcaption}
\usepackage{tabularx}
\usepackage{multirow}
\usepackage{dblfloatfix} 
\usepackage{pifont}
\usepackage{tablefootnote}
\usepackage{xcolor}
\usepackage{thmtools}
\usepackage{thm-restate}
\usepackage{tikz}
\usepackage{hyperref}

\usepackage[normalem]{ulem}

\let\labelindent\relax
\usepackage{enumitem}

\newcommand{\cmark}{\ding{51}}%
\newcommand{\xmark}{\ding{55}}%
    
\newtheorem{theorem}{Theorem}[section]
\newtheorem{lemma}[theorem]{Lemma}
\newtheorem{definition}{Definition}[section]
\newtheorem{corollary}{Corollary}[section]
\DeclareMathOperator{\proj}{proj}
\floatname{algorithm}{Procedure}

\declaretheorem[name= ~,Refname={Assumption, Assumptions}]{assumption} 

\def\BibTeX{{\rm B\kern-.05em{\sc i\kern-.025em b}\kern-.08em
    T\kern-.1667em\lower.7ex\hbox{E}\kern-.125emX}}
\newcommand\copyrighttext{%
  \footnotesize  This paper will be published in IEEE Transactions on Automatic Control.\\ \textcopyright2024 IEEE. Personal use of this material is permitted.
  Permission from IEEE must be obtained for all other uses, in any current or future
  media, including reprinting/republishing this material for advertising or promotional
  purposes, creating new collective works, for resale or redistribution to servers or
  lists, or reuse of any copyrighted component of this work in other works.}
\newcommand\copyrightnotice{%
\begin{tikzpicture}[remember picture,overlay]
\node[anchor=south,yshift=-3pt] at (current page.south) {\fbox{\parbox{\dimexpr\textwidth-\fboxsep-\fboxrule\relax}{\copyrighttext}}};
\end{tikzpicture}%
}

\begin{document}

\title{Optimal Linear Filtering for\\Discrete-Time Systems with\\Infinite-Dimensional Measurements}
\author{Maxwell M. Varley, Timothy L. Molloy, and Girish N. Nair
\thanks{This work received funding from the Australian Government, via grant AUSMURIB000001 associated with ONR MURI grant N00014-19-1-2571. }
\thanks{M. M. Varley and G. N. Nair are with the Department of Electrical and Electronic Engineering, University of Melbourne, Parkville, VIC, 3010, Australia. (emails: varleym@student.unimelb.edu.au, gnair@unimelb.edu.au)}
\thanks{T. L. Molloy is with the CIICADA Lab, School of Engineering, Australian National University, Canberra, ACT 2601, Australia (e-mail: timothy.molloy@anu.edu.au)}
}

\maketitle
\copyrightnotice

\begin{abstract}
Systems equipped with modern sensing modalities such as vision and lidar gain access to increasingly high-dimensional measurements with which to enact estimation and control schemes.
In this article, we examine the continuum limit of high-dimensional
measurements and analyze state estimation in linear time-invariant systems with infinite-dimensional measurements but finite-dimensional states, both corrupted by additive noise. 
We propose a linear filter and derive the corresponding optimal gain functional in the sense of the minimum mean square error, analogous to the classic Kalman filter. By modeling the measurement noise as a wide-sense stationary random field, we are able to derive the optimal linear filter explicitly, in contrast to previous derivations of Kalman filters in distributed-parameter settings. 
Interestingly, we find that we need only impose conditions that are finite-dimensional in nature to ensure that the filter is asymptotically stable.
The proposed filter is verified via simulation of a linearized system with a pinhole camera sensor.

\end{abstract}


\section{Introduction}

\IEEEPARstart{M}{any} systems involve sensors that produce high-dimensional data. These types of systems occur in a wide range of fields, such as modeling economic growth \cite{Belloni2011}, the spiking patterns of neurons in biological systems \cite{Sepulchre}, and vision-based sensing in robotics \cite{Corke}.
In real-time settings with limited computational power, this data is typically converted to a lower-dimensional representation for timely processing, e.g. by extracting a limited number of key features or restricting to a subset of the data.

In this work, we propose to head in the opposite direction, and treat high-dimensional observations as infinite-dimensional random fields on a continuous index set. Our motivation for this approach comes from physics, where it is  often more convenient to treat a complex physical system as a continuum, rather than as a large collection of discrete elements. With the right analytical tools,  a continuous representation can simplify analysis and offer insights into the behaviour of the underlying high-dimensional system.

A system with infinite-dimensional observations does come with its own challenges however, which require the use of different analytical tools and lead to a new perspective compared to its finite-dimensional counterpart. The characterization of infinite-dimensional measurement noise involves stochastic fields, which must be handled carefully to maintain rigour. An example of this is the need for a generalization of the Fubini-Tonelli theorem to wide-sense stationary random fields, which we will present in Theorem \ref{Fub-Ton-Custom}.

Another important distinction is that in a standard finite-dimensional Kalman filter setup, an optimal gain is given by solving a matrix equation, and existence and uniqueness conditions arise naturally from standard linear algebra results. The infinite-dimensional observation space, however, results in an integral equation. We will show in Section \ref{sec:FilterDerivation} (under certain regularity assumptions) that this can be solved via Fourier methods and that the existence and uniqueness conditions can, surprisingly, still be given by examining properties of finite-dimensional objects. This gives us access to a closed-form continuous representation of the optimal gain, which is amenable to Fourier analysis.

Furthermore, we note that mean square stability of finite-dimensional systems often requires stabilizability and detectability conditions dependent on the state transition matrix, input matrix, and output matrix. For systems with infinite-dimensional measurements, it turns out that for mean square stability of our filter we require (sufficient) notions of stabilizability and detectability that take into account the noise covariance matrix, which is not the case for finite-dimensional systems. This is demonstrated in Section \ref{sec:Stability}.

\subsection{Literature Review}
Infinite-dimensional representations are widely used in the study of {\em distributed-parameter systems}  - see e.g. \cite{Morris,Curtain,Bensoussan}. In such systems, the states and measurements are represented as functions both of time and a continuous-valued index taking values in some domain $\mathcal{D}$. The models may be stochastic \cite{Tzafestas1968, Omatu} or deterministic \cite{Morris}, in either discrete or continuous-time \cite{Falbinf}.  
A critical high-level survey of the approaches, claims, and results in this area has been carried out in \cite{Curtain}, and a broader treatment of the estimation and control of distributed-parameter systems can be found in \cite{Morris,Bensoussan,Emirsajłow2021,ZwartCurtain}. An overview of various observer design methodologies for linear distributed-parameter systems is also given in the survey \cite{2011Survey}. Although this survey does not discuss stochastic systems, the references therein are classified as either early lumping models, where the system is discretized before observer design, and late lumping models, where the system is discretized after observer design. In all the works discussed in \cite{2011Survey}, both the state and measurement spaces are considered infinite-dimensional.

In contrast with the rich body of work above, in this article we examine systems with infinite-dimensional measurements but finite-dimensional states.  complementary problem setup occurs in \cite{inverseproblem}, although in that work the system is assumed to have an infinite-dimensional state space and a finite-dimensional measurement space, which sidesteps the issue of calculating the inverse of an infinite-dimensional object, an issue we will discuss further in Section \ref{subsec:comparisons}. The work in \cite{inverseproblem} uses an orthogonal projection operator to reduce the state to a finite-dimensional subspace and proceeds to establish upper bounds on the resulting discretization error by way of sensitivity analysis of a Ricatti difference equation. A more common approach is the aforementioned early lumping method, where discretization is performed at the start to produce linear models with high-but-finite-dimensional states and measurements. This is the approach taken by \cite{KalmanWithNumericalGPs}, where a Kalman filter is designed for a finite-dimensional subset of the infinite-dimensional measurement, corrupted by the additive Gaussian noise that is spatially white. In this article, we do not perform any such lumping, early or late, and our analysis remains firmly in an infinite-dimensional setting.

See Table \ref{tab:prevworknew} for a concise summary of early derivations of linear distributed-parameter filters, as well as more recent work. The papers therein employ a number of different approaches, including Green's function \cite{Thau1969}, orthogonal projections \cite{Tzafestas1968}, least squares arguments \cite{Tzaf1973}, the Wiener-Hopf method \cite{Omatu}, Bayesian approaches \cite{Tzafestas72}, and $\mathcal{H}_2$ estimation \cite{Morris2020}, to name a representative handful. The derivation in this article is closest to the orthogonal projection argument given in \cite{Tzafestas1968}, although here we deal with a discrete-time system. It should also be noted that \cite{Tzafestas1968} assumes that the observation noise is both temporally and spatially white, whereas here we do not require the latter property. We do, however, introduce the assumption that the observation noise is stationary, which allows the novel derivation of a closed form expression of the optimal gain in Section \ref{subsec:OptimalFilter}. This form has not been presented in any of these previous works.

The model that we propose is relevant to systems where
high-dimensional measurements, for instance from vision, radar, or lidar sensors, are used to estimate a low-dimensional state, such as the position and velocity of an autonomous agent or moving target. To the best of the authors' knowledge, such systems have not been the focus of previous work in distributed parameter filtering. 

\begin{table}[t] 
    \begin{center}
    \caption{Properties and assumptions of previous work in state estimation of distributed-parameter models. The final column describes the work presented in this article.}
    \label{tab:prevworknew}
    \resizebox{\columnwidth}{!}{\begin{tabular}{l c c c c c c c c c c}
        \hline
         \footnotesize \textbf{Reference} & \cite{Tzafestas1968} & \cite{Thau1969} & \cite{Meditch1971} & \cite{Tzafestas72} & \cite{Tzaf1973} & \cite{Omatu} & \cite{BensoussanConference} & \cite{Morris} & \cite{Morris2020} & Ours \\
         \hline
         Finite-Dim. State Space &\xmark & \xmark& \xmark&\xmark &\xmark &\xmark &\xmark & \xmark & \xmark & \cmark \\
         Discrete Time &\xmark & \cmark &\xmark & \cmark &\xmark & \cmark &\xmark & \xmark & \xmark & \cmark \\
         Noise is Stationary on Domain \(\mathcal{D}\) &\xmark &\xmark &\xmark &\xmark &\xmark &\xmark &\xmark & \xmark & \xmark & \cmark \\
         Explicit Optimal Gain &\xmark &\xmark &\xmark & \xmark&\xmark &\xmark &\xmark & \xmark & \xmark & \cmark \\
         \hline
    \end{tabular}}
    \end{center}
\end{table}

\subsection{Contributions}
Our key contribution is the derivation of an exact solution for the optimal linear filter for  discrete-time systems with measurement noise and process noise modeled  by random fields and random vectors respectively (Theorem \ref{TheoremOptimal}, Procedure \ref{alg:kalman}). The assumptions of finite state dimension, and that the measurement noise is stationary on domain $\mathcal{D}$, enable us to derive the optimal gain functional explicitly in terms of a multi-dimensional Fourier transform on the measurement index set $\mathcal{D}$.

This is in contrast to previous works in which the optimal gain is given in terms of a distributed-parameter inverse operator that is either not defined or defined implicitly \cite{Omatu}, \cite{Tzafestas72}. Moreover, this inverse can lead to implementation difficulties and is not always guaranteed to exist: these issues are discussed in further detail in Section \ref{subsec:comparisons}. In comparison, our approach allows us to give sufficient conditions for the system to have a well-defined optimal gain function (Section \ref{sec:FilterDerivation}): in rough terms, the effective bandwidth of the Fourier transform of the measurement kernel should be smaller than that of the noise covariance kernel.\footnote{The term bandwidth here refers to the effective range of ``spatial'' frequencies present in a measurement, not temporal frequencies.}

An advantage of this continuum limit approach is that it enables us to make precise statements about the filters behaviour in the limit when measurement resolution approaches infinity. This paves the way for further analysis, as various finite-point approximation schemes may be formulated and the error compared with the asymptotic error of the optimal infinite-resolution filter in a principled and quantitative manner. Further discussion and analysis of these finite-point approximation schemes will be addressed in future work.

This article contains five significant extensions in comparison to our preliminary work in the conference paper \cite{VarleyACC}. 

The first major development is that this work, by way of a Hilbert space framework, presents necessary and sufficient conditions which guarantee the existence and uniqueness of the optimal gain function. Our previous work \cite{VarleyACC} employed a calculus of variations approach to yield only necessary conditions on the optimal gain function.

The second development is the extension of our proposed algorithm to a $d$-dimensional, not just scalar, measurement domain $\mathcal{D}$, allowing much richer sensor data to be modeled without obscuring the underlying spatial correlations. Examples include camera images ($d=2$), as well as high-resolution lidar and magnetic resonance imaging (MRI)  ($d=3$). This extension entails the use of random fields and Fourier transforms with $d$-dimensional frequency, as opposed to standard random processes and  Fourier transforms on a scalar axis. This has also allowed us to present simulation results that involve two-dimensional images, as opposed to the simpler scalar functions previously presented in \cite{VarleyACC}.

The third development is a weakening of the assumptions on the observation noise. The filter provided in \cite{VarleyACC} was restricted to systems with Gaussian observation noise, whereas the formulation presented here loosens this assumption and instead only requires that the noise be stationary.

The fourth development is that this work clearly defines explicit conditions such that the presented arguments are rigorous, whereas the derivations in \cite{VarleyACC} were of a more preliminary and exploratory nature. This work provides a list of assumptions such that our derived filter is well-defined, and that each step in the derivation is valid.

The fifth and final development is that we explicitly impose {\em stabilizability} and {\em detectability} conditions on the system (Section \ref{sec:StabDetect}). The detectability condition in particular involves the measurement noise covariance function, as well as the deterministic dynamical matrix and measurement function. This strongly contrasts with other works in infinite-dimensional systems theory, which impose conditions of \emph{strong} or \emph{weak} observability \cite{Fuhrmann1973, Megan} that are agnostic to noise.

These extensions yield a clearer picture of the relationship between the system and the optimal filter, and provide a more generalized solution for the optimal gain. 

\begin{figure*}[t]
    \centering
    \begin{subfigure}{0.33\textwidth}
        \includegraphics[scale=0.33]{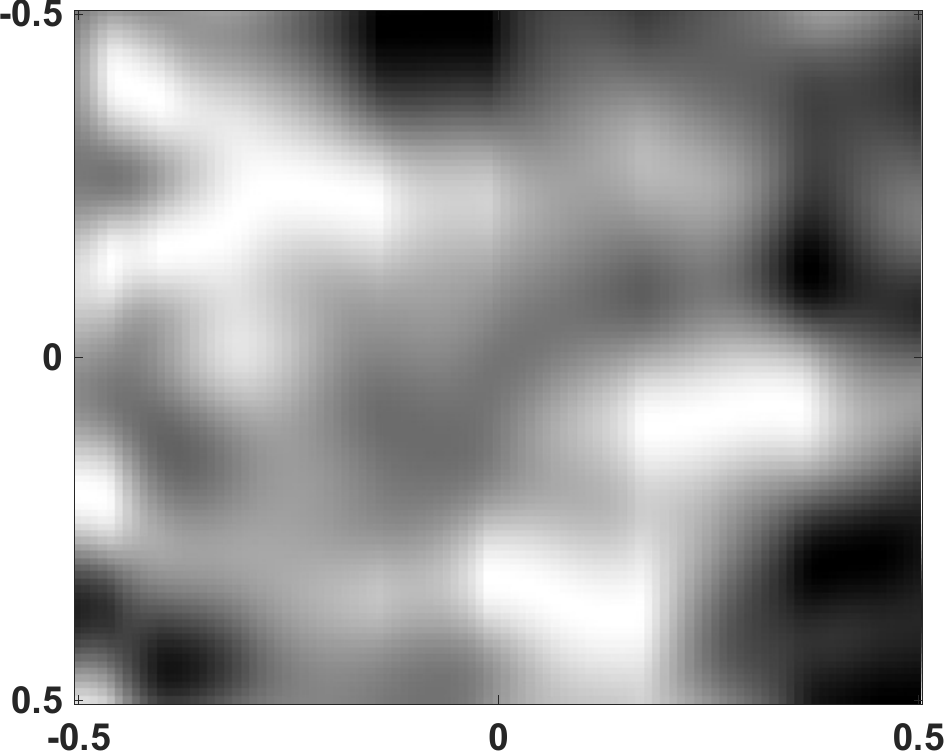}
        \caption{Length scale \(\ell=0.1\)}
    \end{subfigure}\begin{subfigure}{0.33\textwidth}
        \includegraphics[scale=0.33]{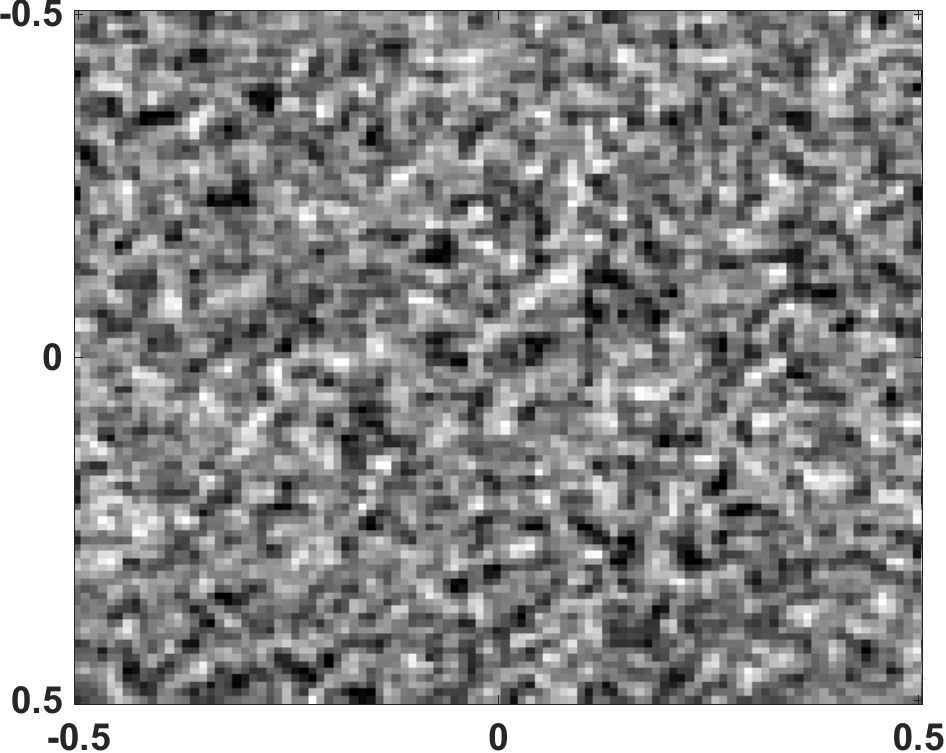}
        \caption{Length scale \(\ell=0.01\)}
    \end{subfigure}\begin{subfigure}{0.33\textwidth}
        \includegraphics[scale=0.33]{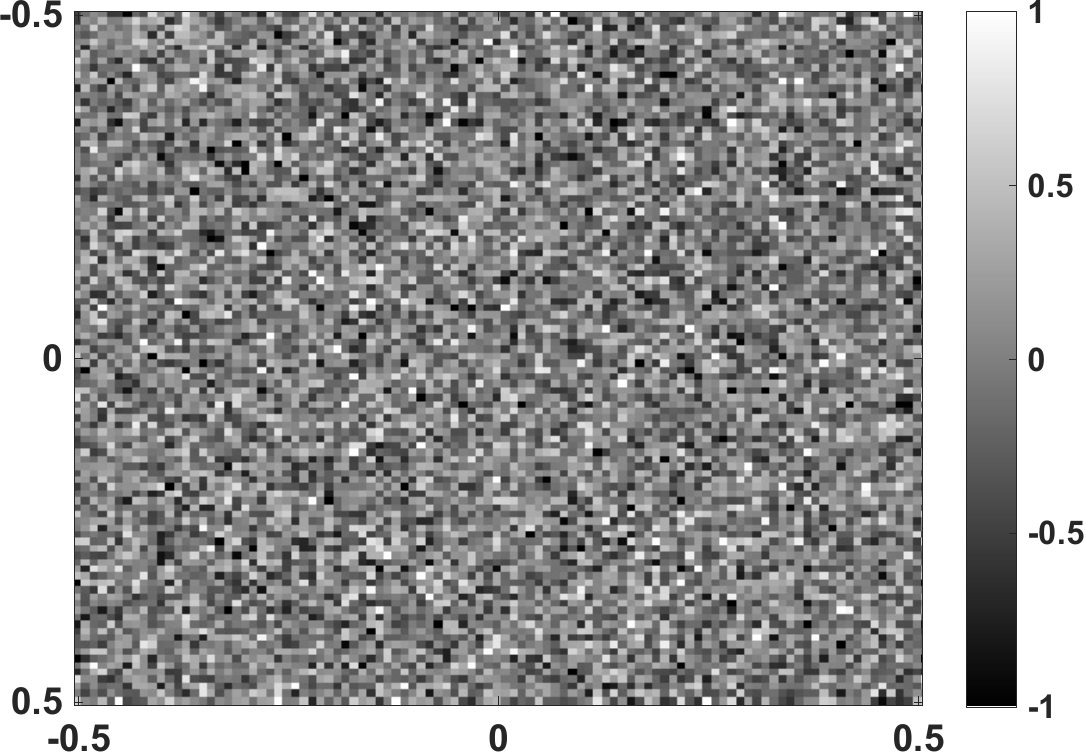}
        \caption{Length scale \(\ell=0.001\)}
    \end{subfigure}
    \caption{Three normalized stationary random fields, spatially discretized into a \(100\times 100\) grid over the domain \([-0.5,0.5]^2\). Each field is zero-mean with a squared exponential covariance function \(R(i,i')\propto e^{-\|i-i'\|_2^2(2\ell^2)^{-1}}\). The length scale \(\ell\) determines how closely each element is influenced by its neighbors, with a larger length scale leading to a stronger correlation. Note that small length scales are often used to approximate ideal white noise.}
    \label{fig:GreyScaleNoise}
\end{figure*}
\subsection{Organization}
This article is organized as follows. Section \ref{sec:Preliminaries} outlines the relevant definitions, theorems, and conditions relating to the Hilbert space of random variables and random fields, as well as the multi-dimensional Fourier transform and notions of system stabilizability and detectability. Section \ref{sec:ModelDefinition} gives our model definition and discusses some concrete examples which our model could be feasibly used to represent. This section also summarizes the key assumptions needed to ensure validity of our analysis, and the motivation behind each assumption. Section \ref{sec:FilterDerivation} derives necessary and sufficient conditions for the optimal gain of the proposed filter (\ref{subsec:OptimalityConditions}), before presenting a unique gain function which satisfies these conditions \ref{subsec:ExplicitSol}. Section \ref{subsec:OptimalFilter} will present the full filter equations and stability of this filter will be discussed in Section \ref{sec:Stability}. Section \ref{sec:Implementation} will give a procedure to implement this filter, with Section \ref{subsec:complexity} discussing computational complexity and Section \ref{subsec:comparisons} drawing comparisons to existing results. Finally, Section \ref{sec:Simulations} presents a simulated system and analyzes the empirical performance of the filter compared to its expected performance.

\section{Preliminaries} \label{sec:Preliminaries}
\subsection{Hilbert Space of Random Variables}
In this work, we denote by \(\mathcal{H}\), the Hilbert space of all real-valued random variables with finite second moments \cite[Ch. 20]{Fristedt}. 
The inner product of any two elements \(u,v\in\mathcal{H}\) is defined by \(\langle u,v\rangle=E[uv].\) An inner product operating on two random vectors \(u,v\) in the Hilbert space \(\mathcal{H}^n\) of random \(n\)-vectors is defined by \(\langle u,v\rangle_\mathcal{H}=E[u^\top v]\). The derivation used in this article involves defining the state of our observed system as a vector residing in a Hilbert space and then selecting an error-minimizing estimate that resides in a subspace. The following theorem \cite[p.51, Th. 2]{vectoropt} is fundamental to this approach.
\begin{theorem}[Hilbert Projection Theorem]\label{HilbertProjTheorem}
Let \(\mathcal{G}\) be a Hilbert space
and \(M\) a closed subspace of \(\mathcal{G}.\) Corresponding to any vector \(x\in \mathcal{G}\), there exists a
unique vector \(m_0\in M\) such that \(||x - m_0|| \leq ||x - m|| \ \forall \ m\in M.\)
Furthermore, a necessary and sufficient condition that \(m_0 \in M\) is the unique minimizing vector is that \((x-m_0) \perp M.\)
\end{theorem}

\subsection{Random Fields}
This article models measurement noise as an additive random field and some care is required when manipulating these objects. Rigorous and extensive analysis of these fields can be found in a number of textbooks \cite{adlerrandom}, \cite{doobstochastic}, \cite{janson}. A brief summary of relevant definitions is given below.

Let \((\Theta,\mathcal{F}, P)\) be a probability space and \(\mathbb{R}^d\) a Euclidian space endowed with Lebesgue measure. A random field is a function \(v:\mathbb{R}^d\times \Theta\rightarrow \mathbb{R}^m\) which is measurable with respect to the product measure on \(\mathbb{R}^d\times\Theta\).\footnote{For concision, the dependence on \(\theta\) will usually be suppressed.} A centered w.s.s. (wide-sense stationary) random field has the  properties that i) \(E[v(i)] = 0\) for all \(i\); ii) the random vector \(v(i) \in\mathcal{H}\) for all \(i\); and iii) the covariance function \(R\) only depends on the difference between index points, so that for any \(i_1,i_2\in\mathbb{R}^d, \ E[v(i_1)v^\top(i_2)]=R(i_1-i_2)\). Such a field that is stationary over a domain \(\mathcal{D}\) is often referred to as \(\mathcal{D}-\)stationary noise. In the special case where the covariance is also independent of direction, the noise is said to be isotropic. The stationarity condition arises in situations
where the noise statistics possess translational invariance, in other words, the noise statistics look the same everywhere throughout the domain. Some sources of noise may also be locally stationary, in that they are
closely approximated by stationary processes over finite
windows. Spatial stationarity is a common assumption in a
number of diverse contexts such as image denoising \cite{Corke, Cox, Fehrenbach}, radio signal propagation \cite{stationaryRadio}, and geostatistics \cite[Ch. 74]{Geostat}. Figure \ref{fig:GreyScaleNoise} gives three example realizations of a stationary random field with a squared exponential covariance function and varying length scales.

\subsection{The Multi-Dimensional Fourier Transform}
A critical step in this article involves the transformation of an integral equation to an algebraic equation, facilitated by the multi-dimensional Fourier transform, which we now define.
\begin{definition}[Multi-Dimensional Fourier Transform]\cite[Ch. 1]{stein_fourier}\label{def:fourier}
The multi-dimensional Fourier transform of any \(f\in L_1(\mathbb{R}^d,\mathbb{R})\) is given by \(\mathcal{F}\{f\}\in L_\infty(\mathbb{R}^d,\mathbb{R})\) and is defined as
\begin{align}
    (\mathcal{F}\{f\})(\omega)\triangleq \int_{\mathbb{R}^d}f(i)e^{-2\pi j\omega\cdot i}di,\nonumber
\end{align}
where \(\omega\) is a \(d\)-dimensional vector of spatial frequencies. We will denote the Fourier transform of any function \(f\) by \(\Bar{f}.\) The inverse Fourier transform is defined by
\begin{align}
    (\mathcal{F}^{-1}\{f\})(i)=\int_{\mathbb{R}^d}\Bar{f}(\omega)e^{2\pi ji\cdot \omega}d\omega.\nonumber
\end{align}
This form of the inverse generally only holds when \(\Bar{f}(\omega)\) is also absolutely integrable \cite{follandfourier}.
\end{definition}

\subsection{Stabilizability and Detectability}\label{sec:StabDetect}
This article will employ notions of stabilizability and detectability, which we define as follows:
\begin{definition}[Stabilizability]{\cite[p. 342]{AndersonandMoore}}\label{def:stabilizable}
The matrix pair \((A,B)\) is \textit{stabilizable} if there exists a matrix \(C\) such that
\begin{align}
    \rho(A+BC^{\top})<1,\nonumber
\end{align}
where \(\rho(\cdot)\) is the spectral radius.
\end{definition}
\begin{definition}[Detectability]{\cite[p. 342]{AndersonandMoore}}\label{def:detectable}
The matrix pair \((A,B)\) is \textit{detectable} if \((A^\top, B)\) is stabilizable.
\end{definition}

\section{Problem Formulation and Assumptions}\label{sec:ModelDefinition}
Consider the linear discrete-time system
\begin{align}
    x_{k+1}&=Ax_k+w_k, \label{eq:statedynamics}
\end{align}
with state \(x_k\in \mathbb{R}^n\), additive process noise \(w_k\in \mathbb{R}^n\), state transition matrix \(A \in \mathbb{R}^{n\times n}\), and time index \( k\in \mathbb{N}\). We omit a control input for clarity; the results derived in this article are still valid with minor changes if a control input is present and available to the estimator.
The state is measured via an infinite-dimensional measurement field defined on domain \(\mathbb{R}^d\), that is,
    \begin{align}\label{eq:observation}
        z_k(i)=\gamma(i)x_k+v_k(i).
    \end{align}
    Note that the measurement domain \(i\in\mathbb{R}^d\), measurement \(z_k:\mathbb{R}^d\rightarrow \mathbb{R}^m\), measurement noise \(v_k:\mathbb{R}^d\rightarrow \mathbb{R}^m\), and measurement function \(\gamma:\mathbb{R}^d\rightarrow \mathbb{R}^{m\times n}\).
    
    The schema presented in \eqref{eq:statedynamics} and \eqref{eq:observation} are well-suited to systems where the state can be represented by a low-dimensional model, but the observations are represented by a very high or infinite-dimensional model. A major area where such systems typically hold is robotics, in the context of Simultaneous Localization and Mapping (SLAM). The state to be estimated is often the pose of a mobile robot, which could be described, for example, as a 9-vector that represents the robot position, velocity, and acceleration in \(\mathbb{R}^3\). The state updates are generally linear, obeying Newtonian dynamics. The measurement sensors, on the other hand, can be of a much higher dimensionality.
    
    When lighting conditions are adequate, optical cameras provide a rich source of data, and even a relatively cheap 10 megapixel camera will provide a ten million dimensional image every frame. In these conditions, vision-based localization can be particularly suited to GPS-denied areas, such as extra-terrestrial environments \cite{spacenav}. In the case of a colour camera, the measurement equation described \eqref{eq:observation} would apply with a \(3\)-vector \((m=3)\) describing the RGB intensity of each pixel, defined on a \(2-\)dimensional image domain \((d=2)\).
    
    Under low or no lighting conditions, where optical cameras do not provide sufficient data, lidar is often employed. This setup has applications in the mining sector \cite{LHDTopo, LHDOverview}, where underground tunnels or other "visually degraded" environments must be traversed \cite{lowlightSLAM}. Although not generally as ``pixel-dense'' as optical cameras, the commonly used SICK LMS511 is capable of providing up to 4560 samples per scan, and higher-end lidar modules can provide upwards of a million samples per scan. This schema is also suited to target-tracking problems, where the state to be estimated now represents a moving body, and the measurement sensors are often radar, lidar, optical, or an ensemble thereof \cite{RadLidVis}.
    
    We assume the additive measurement noise in \eqref{eq:observation} is a w.s.s. centered random field with bounded covariance function \(R\), so that each element of \(z(i)\) is in \(\mathcal{H}\) for all \(i\in \mathbb{R}^d\). We also assume that \(\gamma\) is bounded and absolutely integrable over domain \(\mathbb{R}^d\) under the Frobenius or other matrix norm.
    
The process noise and measurement noise are such that $\forall j,k\in \mathbb{N}$ and $i,i'\in \mathbb{R}^d$,
\begin{align}
    E[w_k]&=0\nonumber\\
    E[v_k(i)]&=0\nonumber\\
    E[v_k(i) w_j^\top] &= 0\nonumber\\
    E[w_kw_j^\top]&=Q\delta_{j-k}, &&Q\in \mathbb{R}^{n\times n}\nonumber\\
    E[v_k(i)v_j^\top(i')]&=R\big(i-i'\big)\delta_{j-k}, &&R\big(i-i'\big)\in \mathbb{R}^{m\times m}.\label{eq:obsnoisecov}
\end{align}
Here \(Q\) is a positive-definite matrix, \(\delta_{j-k}\) is the discrete impulse, and \(R\) is bounded and absolutely integrable. 

We will examine the performance of a linear filter of the form
\begin{align}
    \begin{cases}
    \hat{x}_{k}&=A\hat{x}_{k-1}+K_k[z_k-\hat{z}_k] \\
    \hat{z}_k(i)&=\gamma(i)A\hat{x}_{k-1},
    \end{cases} \label{eq:predictor}
\end{align}
under a mean square error (MSE) criterion \(E[\|x_k-\hat{x}_k\|_2^2]\). Here \(K_k\) is a linear mapping from the space of \(\mathbb{R}^m\)-valued fields on \(\mathbb{R}^d\) to \(\mathbb{R}^n\). The innovation term we will denote by \(s_k\triangleq z_k-\hat{z}_k\). Motivated by the Riesz representation theorem\cite[p. 188]{Kreyszig}, we assume the integral form
\begin{align}
    K_ks_k \triangleq \int_{\mathbb{R}^d} \kappa_k (i)s_k(i) di,\label{eq:kappaform}
\end{align}
 where \(\kappa_k :\mathbb{R}^d\rightarrow\mathbb{R}^{n\times m}\) is an optimal gain function to be determined.
\\We now present the assumptions we need, but before we do so, we introduce the following useful terms.
    \begin{align}
        f(i)&\triangleq\mathcal{F}^{-1}\{\bar{\gamma}^\top \bar{R}^{-1}\}\in \mathbb{R}^{n\times m}, \nonumber\\ S&\triangleq\int_{\mathbb{R}^d} f(i)\gamma(i)\,di \in \mathbb{R}^{n\times n}.\label{def:S}
    \end{align}
These terms are well-defined given the assumptions below, and we will also show in Section \ref{sec:Stability} that \(S\) is symmetric and positive semi-definite, and hence there exists a unique symmetric positive semi-definite matrix \(G\) such that \(S=GG\)\cite[p. 439]{Horn13}. This matrix \(G\) will be termed the principal square root of \(S\).
A summary of the key assumptions stated for our model is now presented.
\begin{restatable}{assumption}{assumgamma}\label{ass:gamma}
    \(\gamma\in L_1(\mathbb{R}^d,\mathbb{R}^{n\times m})\cap L_\infty(\mathbb{R}^d,\mathbb{R}^{n\times m})\)
\end{restatable}
\begin{restatable}{assumption}{assumR}
    \label{ass:R} \(R\in L_1(\mathbb{R}^d,\mathbb{R}^{m\times m})\cap L_\infty(\mathbb{R}^d,\mathbb{R}^{m\times m})\)
\end{restatable}
\begin{restatable}{assumption}{assumbarR}
    \label{ass:barR} \(\bar{R}(\omega)\in \mathbb{R}^{m\times m}\) is invertible for almost all \(\omega\)
\end{restatable}
\begin{restatable}{assumption}{assumRGamma}
    \label{ass:RGamma} \(\bar{\gamma}^\top\bar{R}^{-1}\in L_1(\mathbb{R}^d, \mathbb{R}^{n\times m})\cap L_2(\mathbb{R}^d, \mathbb{R}^{n\times m})\) and has an inverse fourier transform in \(L_1(\mathbb{R}^d,\mathbb{R}^{n\times m})\)
\end{restatable}
\begin{restatable}{assumption}{assumkappa}
    \label{ass:kappa} \(\kappa_k\in L_1(\mathbb{R}^d,\mathbb{\mathbb{R}}^{n\times m})\cap L_2(\mathbb{R}^d,\mathbb{\mathbb{R}}^{n\times m}) \ \forall k\)
\end{restatable}
\begin{restatable}{assumption}{assumstability}
    \label{ass:stability} The matrix pair \((A,Q)\) is stabilizable.
\end{restatable}
\begin{restatable}{assumption}{assumdetectability}
    \label{ass:detectability} The matrix pair \((A,G)\) is detectable.
\end{restatable}
Assumption \ref{ass:gamma} and assumption \ref{ass:R} will be used in Section \ref{subsec:ExplicitSol} to ensure that an appropriate Fourier transform of \(\gamma\) and \(R\) is well-defined. Assumption \ref{ass:gamma} is also needed to extend the Fubini-Tonelli theorem to wide-sense-stationary random fields, as shown in Appendix \ref{app:Fub-Ton}. Assumption \ref{ass:barR} and assumption \ref{ass:RGamma} will be needed in Section \ref{subsec:ExplicitSol} to ensure the existence of an explicit form of the optimal gain function \(\kappa\). Assumption \ref{ass:kappa} is imposed to ensure that the integral \eqref{eq:kappaform} is well-defined by Theorem \ref{Fub-Ton-Custom}. Whilst there are notions of observability in infinite-dimensional systems theory (cf.~\cite{Fuhrmann1973, Megan}), in this article we do not require them.
Instead, we shall assume that \((A,Q)\) is {\em stabilizable} (assumption \ref{ass:stability}) and that \((A,G)\) is {\em detectable} (assumption \ref{ass:detectability}) in the sense of Definition \ref{def:stabilizable} and Definition \ref{def:detectable}, where \(G\) is a matrix-valued functional of the measurement function \(\gamma\) and the measurement noise covariance \(R\). This will ensure that the filter error covariance matrix converges as time grows. It should be noted that traditional detectability conditions do not typically depend on properties of the measurement noise \cite{AndersonandMoore,Fuhrmann1973, Megan}.

We verify that the filter integral \eqref{eq:kappaform} is well-defined by adapting the classical Fubini-Tonelli theorem \cite[Cor. 13.9]{schilling_2005} to w.s.s. random fields as follows:
\begin{theorem}[Fubini-Tonelli for w.s.s. random fields] \label{Fub-Ton-Custom}
Let \(v:\mathbb{R}^d\times \Theta\rightarrow\mathbb{R}^m\) be a w.s.s. random field defined on the probability space \((\Theta,\mathcal{F},P)\), with \(\theta\in \Theta\) representing a point in the sample space \(\Theta\), and with bounded covariance function \(\Sigma\). Let \(s:\mathbb{R}^d\times \Theta\rightarrow\mathbb{R}^m\) be a random field given by \(s\triangleq v+\gamma x\), where \(\gamma\in L_\infty(\mathbb{R}^d,\mathbb{R}^{m\times n})\) and \(x\) is a random vector with elements in \(\mathcal{H}\) and covariance matrix \(P\). Also let \(\kappa:\mathbb{R}^d\rightarrow\mathbb{R}^{n\times m}\) be a matrix-valued function, absolutely integrable under the Frobenius matrix norm. Each component of the vector-valued function \(\kappa(\cdot)s(\cdot,\cdot)\) is assumed to be measurable with respect to the product \(\sigma\)-algebra \(\mathcal{L}\times\mathcal{F}\), where \(\mathcal{L}\) is the \(\sigma\)-algebra corresponding to the \(d\)-dimensional Lebesgue measure. Then
    \[\int_{\mathbb{R}^d}\int_\Theta \|\kappa(i)s(i,\theta)\|_1P(d\theta) di<\infty,\]
\end{theorem}
and furthermore 
\[\int_\Theta\int_{\mathbb{R}^d}\kappa(i)s(i,\theta)diP(d\theta)=\int_{\mathbb{R}^d}\int_\Theta \kappa(i)s(i,\theta)P(d\theta)di. \]
The proof of Theorem \ref{Fub-Ton-Custom} is given in Appendix \ref{app:Fub-Ton}

\section{Optimal Linear Filter Derivation} \label{sec:FilterDerivation}

We wish to find the optimal state estimate \(\hat{x}_k\), where \(\hat{x}_k\) is some bounded linear functional of all measurements up to time \(k.\) Optimality in this work will refer to the estimate which minimizes the mean square error \(E[\|x_k-\hat{x}_k\|^2_2]\). This random \(n\)-vector minimization problem is equivalent to \(n\) random scalar minimization problems, as minimization of each component norm will minimize the entire vector norm.
Each component of \(\hat{x}_k\) is then the orthogonal projection of the corresponding component of \(x_k\) onto the space of all scalar-valued linear functionals of the preceding measurements.
Let \(x_k^l\in \mathcal{H}\) denote the \(l^{th}\) component of the state vector \({x}_k\). Also let our observation at time \(k\) be denoted by the random \(m\)-vector valued function \({z}_k:\mathbb{R}^d\rightarrow \mathbb{R}^m\), assumed measurable with respect to the product measure on \(\mathbb{R}^d\times\Theta\), and the 
 set of all possible measurements at time \(k\)  be \(\mathcal{Z}_k\). Section \ref{subsec:OptimalityConditions} will state the optimality conditions of the filter, and Section \ref{subsec:ExplicitSol} will derive an explicit solution that satisfies these conditions. A brief reference of the spaces and variables used in the following derivation is given in Table \ref{tab:Spaces} and Table \ref{tab:variables} respectively.
\begin{table}[ht]
    \centering
    \caption{Spaces used in Section \ref{sec:FilterDerivation} derivations.}
    \label{tab:Spaces}
    \begin{tabular}{|p{0.1\linewidth}|p{0.80\linewidth}|}
    \hline
         \textbf{Symbol}&\textbf{Description}  \\
         \hline
          \(\mathcal{H}\)&Hilbert space of random variables with zero mean and finite second moment.\\ 
          \hline
          \(\mathcal{Z}_k\)&Space of all possible measurements \(z_k(\cdot)\).\\ 
          \hline
          \(M_{k-1}\)&Space of all possible elements of \(\mathcal{H}\) expressible via \eqref{eq:Mdef}.\\
          \hline
          \(M_{k}^-\)&Space of all possible elements of \(\mathcal{H}\) expressible via \eqref{eq:Mplusdef}.\\ 
          \hline
          \(\mathcal{S}_k\)&Space of all possible elements of \(\mathcal{H}\) expressible via \eqref{eq:Sdef}.\\
          \hline
    \end{tabular}
\end{table}

 \begin{table}[ht]
    \centering
    \caption{Variables used in Section \ref{sec:FilterDerivation} derivations.}
    \label{tab:variables}
    \begin{tabular}{|p{0.1\linewidth}|p{0.80\linewidth}|}
    \hline
         \textbf{Symbol}&\textbf{Description}  \\
         \hline
          \(x\)&System state \\
          \hline
          \(z\)& Measurement\\
          \hline
          \(\hat{x}\)&Optimal estimate of \(x\) given past measurements \\
          \hline
          \(\hat{z}\)& Optimal estimate of \(z\) given past measurements\\ 
          \hline
          \(s\)& Innovation vector \(z-\hat{z}\)\\
          \hline
          \(\phi,\alpha\)&Arbitrary vectors in \(L_1(\mathbb{R}^d,\mathbb{R}^{1\times m})\cap L_2(\mathbb{R}^d,\mathbb{R}^{1\times m})\)\\
          \hline
          \(m\)&Arbitrary element expressible via \eqref{eq:Mdef}\\ 
          \hline
          \(\tilde{s}\)&Arbitrary element expressible via \eqref{eq:Sdef}\\ 
          \hline
          \(\beta\)&Optimal post-measurement correction term\\
          \hline
          \(\kappa\)& Optimal kernel which generates \(\beta\)\\
          \hline
          \((\cdot)_k\)&The object \((\cdot)\) at time step \(k\)\\
          \hline
    \end{tabular}
\end{table}
\subsection{Optimality Conditions}\label{subsec:OptimalityConditions}
In this section, we will prove the following key theorem:
\begin{theorem} \label{TheoremOptimal}
    The matrix kernel \(\kappa_k(i)\in L_1(\mathbb{R}^d,\mathbb{R}^{n\times m})\cap L_2(\mathbb{R}^d,\mathbb{R}^{n\times m})\) for \eqref{eq:predictor} is optimal in the mean square error sense if and only if
    \begin{align}
        &\int_{\mathbb{R}^d}\kappa_k(i)\bigg(\gamma(i)P_{k|k-1}\gamma^\top(i')+R(i-i')\bigg)di=P_{k|k-1}\gamma^\top(i').\label{eq:ACCcondition}
    \end{align}
    where \(P_k\) and \(P_{k|k-1}\) denote the posterior covariance matrix and prior covariance matrix of the filter respectively.
\end{theorem}
Before we provide a proof of Theorem \ref{TheoremOptimal}, we must first define some important sets, and projections onto those sets.\\
Consider the set of all possible elements of \(\mathcal{H}\) which can be expressed as a bounded  linear integral functional operating over all previous measurements:
\begin{align}
    M_{k-1}\triangleq \bigg\{m_{k-1}\in \mathcal{H}: m_{k-1}=\sum_{j=1}^{k-1}\int_{\mathbb{R}^d} \phi_j(i)z_j(i) di\bigg\},\label{eq:Mdef}&\\\text{ for some }\phi_j \in L_1(\mathbb{R}^d,\mathbb{\mathbb{R}}^{1\times m})\cap L_2(\mathbb{R}^d,\mathbb{\mathbb{R}}^{1\times m})&.\nonumber
\end{align}
Note that \(M_{k-1}\) is a subspace of \(\mathcal{H}\). We will also assume for the moment that \(M_{k-1}\) is a closed subspace of \(\mathcal{H}\).
Suppose that we are already in possession of the optimal estimate of \(x_k^l\) given the previous measurements. This estimate is the orthogonal projection of \(x_k^l\) onto \(M_{k-1}\),
\begin{align}
    \hat{x}_{k|k-1}^l\triangleq\proj_{M_{k-1}}x_k^l.\nonumber
\end{align}
Given the measurements associated with \(M_{k-1},\; \hat{x}_{k|k-1}^l\) is the best estimate of \(x_k^l\) in the sense that it is the element of \(M_{k-1}\) that is the unique \(\mathcal{H}\)-norm minimizing vector in relation to the \(l^{th}\) component of the true state.\par

Let \(z_k^p\) be the \(p^{th}\) element of \({z}_k\) and \(\hat{z}_k^p\) denote the orthogonal projection of \(z_k^p\) onto \(M_{k-1}\),
\begin{align}
    \hat{z}_k^p(i)\triangleq\proj_{M_{k-1}}z_k^p(i)\quad \forall i.\nonumber
\end{align}
The measurement innovation vector is \(s_k(i)\triangleq z_k(i)-\hat{z}_k(i)\), the \(p^{th}\) entry of which is denoted \(s^p_k(i)\in \mathbb{R}.\) Note that for all \(m_{k-1}\in M_{k-1}\),
\begin{align}\label{eq:orthospaces}
    s^p_k(i)\perp m_{k-1}\quad \forall i.
\end{align}
and let
\begin{align}
    \mathcal{S}_k\triangleq \bigg\{\tilde{s}_k\in \mathcal{H}:\tilde{s}_k=\sum_{p=1}^m\int_{\mathbb{R}^d}\alpha_k^p(i)s_k(i)di\bigg\},\label{eq:Sdef}\\
    \text{for some }\alpha_k^p\in L_1(\mathbb{R}^d,\mathbb{\mathbb{R}}^{1\times m})\cap L_2(\mathbb{R}^d,\mathbb{\mathbb{R}}^{1\times m}).\nonumber
\end{align} 
which is a subspace of \(\mathcal{H}.\) Here \({\alpha}_k\in L_1(\mathbb{R}^d,\mathbb{\mathbb{R}}^{n\times m})\cap L_2(\mathbb{R}^d,\mathbb{\mathbb{R}}^{n\times m}),\) with the \(p^{th}\) row given by \(\alpha_k^p\).

\begin{lemma}[\(\mathcal{S}_k\perp M_{k-1}\)]\label{lemma_2}
Given the previous definitions of \(\mathcal{S}_k,M_{k-1},\) every element of \(\mathcal{S}_k\) is orthogonal to every element of \(M_{k-1}.\)
\end{lemma}
\begin{proof}
    \(\alpha_k^p\in L_1(\mathbb{R}^d,\mathbb{R}^{1\times m})\) and \(s_k^p\) is the sum of a random field with bounded covariance function, and an absolutely integrable function. We may therefore apply Theorem \ref{Fub-Ton-Custom} and exchange integration and expectation operators.
\begin{align}
    \bigg\langle \tilde{s}_k, m_{k-1}\bigg\rangle_\mathcal{H} &=\bigg\langle \sum_{p=1}^m\int_{\mathbb{R}^d}\alpha_k^p(i)s_k(i)di, m_{k-1}\bigg\rangle_\mathcal{H}\nonumber\\ &=E\bigg[\left(\sum_{p=1}^m\int_{\mathbb{R}^d}\alpha_k^p(i)s_k(i)di\right)\cdot m_{k-1}\bigg]\nonumber\\
    &= \sum_{p=1}^m\int_{\mathbb{R}^d}\alpha_k^p(i)E[s_k(i)\cdot m_{k-1}]di\nonumber\\
    &=0\label{eq:ortho}
\end{align}
The final equality following from the orthogonality property stated in  \eqref{eq:orthospaces}.
\end{proof}
Having shown that the spaces \(\mathcal{S}\) and \(M_{k-1}\) are orthogonal, we will now find the projection of the true state \(x_k\) onto the space of all possible elements in \(\mathcal{H}\) that can be expressed by a bounded linear integral functional operating over all previous measurements \textit{and} the current measurement.
Let \(M^-_{k}\) be the set of all possible outputs of linear functionals of any element in \(\mathcal{Z}_k\), 
\begin{align}
    M^-_{k}\triangleq \bigg\{m^-_{k}\in\mathcal{H}:m^-_{k}=\int_{\mathbb{R}^d} \phi(i)z_k(i) di\bigg\},\label{eq:Mplusdef}&\\\text{ for some }\phi \in L_1(\mathbb{R}^d,\mathbb{\mathbb{R}}^{1\times m})\cap L_2(\mathbb{R}^d,\mathbb{\mathbb{R}}^{1\times m}), \ z_k\in \mathcal{Z}_k&.\nonumber
\end{align}
Note that using this definition, it immediately follows that \(M_{k-1}+M^-_{k}=M_{k}\), where the addition of sets is given by the Minkowski sum
\begin{align*}
    U+V=\{u+v|u\in U, v\in V\}.
\end{align*}
We also use the \(\oplus\) operator to refer to a direct sum of subspaces, where any element within the direct sum is uniquely represented by a sum of elements within the relevant subspaces \cite{axler}.
We wish to find the projection of \(x_k^l\) onto the space of \(M_{k-1}+ M^-_{k}\), which is given by
\begin{align}
    \hat{x}^l_k&=\proj_{M_{k-1}+M^-_{k}}x_k^l\nonumber\\
    &=\proj_{M_{k-1}\oplus \mathcal{S}_k}x_k^l\nonumber\\
    &=\proj_{M_{k-1}}x_k^l+\proj_{\mathcal{S}_k}x_k^l\nonumber\\
    &=\hat{x}_{k|k-1}^l+\beta_k^l,\quad\beta_k^l\triangleq\proj_{\mathcal{S}_k}x_k^l.\label{eq:projectionform}
    \end{align}
This sum represents a decomposition of the projected state component given measurements up to time \(k\) as a unique sum of a component within \(M_{k-1}\) and a component orthogonal to \(M_{k-1}.\)
Note that as \(\beta^l_k\in {\mathcal{S}_k}\), it can be formulated as \(\sum_{p=1}^m\int_{\mathbb{R}^d}\kappa_k^{l,p}(i)s_k^p(i)di\), each \(\kappa_k^{l,p}(i)\) being the \(l^{th}\) row, \(p^{th}\) column element of some \({\kappa}_k(i)\in \mathbb{R}^{n\times m}\) which must be determined. This matrix-valued function \({{\kappa}_k}(i)\) is equivalent to our optimal gain. We are now ready to prove Theorem \ref{TheoremOptimal}, which we will do by formulating necessary and sufficient conditions on the gain such that the optimal post-measurement correction term \(\beta_k=\int_{\mathbb{R}^d}\kappa_k(i)s_k(i)\,di\) is generated. The proof of Theorem \ref{TheoremOptimal} is as follows:

\begin{proof}
We will denote the \(l^{th}\) column of \({\kappa}^\top_k(i)\) as \({{\kappa}^l_k}(i)\in \mathbb{R}^{m}.\) By the Hilbert projection theorem, \((x_k^l-\beta^l_k) \perp {\mathcal{S}_k}\implies \langle x_k^l-\beta^l_k, \tilde{s}_k \rangle_\mathcal{H}=0\) for all \(\tilde{s}_k\in \mathcal{S}_k\), where \(\beta^l_k\) is guaranteed to be unique. We then have that \(\langle x_k^l, \tilde{s}_k \rangle_\mathcal{H} = \langle \beta^l_k, \tilde{s}_k \rangle_\mathcal{H}\) which, in conjunction with the fact that \(\langle \hat{x}_{k|k-1}^l, \tilde{s}_k \rangle_\mathcal{H} = 0\), implies
\begin{align}\label{eq:conditionproj}
    \langle x_k^l-\hat{x}_{k|k-1}^l, \tilde{s}_k \rangle_\mathcal{H} = \langle \beta^l_k, \tilde{s}_k \rangle_\mathcal{H}.
\end{align}
This is a necessary and sufficient condition for \(\beta_k^l\) to be the projection of our state component onto the space \(\mathcal{S}_k.\) We now must find the appropriate optimal gain function \(\kappa_k(i)\) such that \(\beta_k\), the corresponding vector generated, satisfies condition \eqref{eq:conditionproj} for all \(l.\) Expanding out the definitions of \(\beta_k^l\) and \(\tilde{s}_k\), interchanging integration and summation, and rearranging, leads to
\small
\begin{align}
        \bigg\langle\mkern-5mu x_k^l-\hat{x}_{k|k-1}^l-\int_{\mathbb{R}^d}\sum_{p=1}^m\kappa_k^{l,p}(i)s_k^p(i)di,\mkern-5mu\int_{\mathbb{R}^d}\sum_{p=1}^m\alpha_k^p(i)s^p_k(i)di\mkern-5mu\bigg\rangle_\mathcal{H}\mkern-14mu=\mkern-2mu0\nonumber
\end{align}
\normalsize
Note that \(\sum_{p=1}^m\kappa_k^{l,p}(i){s}_k^p(i)\) is simply the vector multiplication \({s}^\top_k(i){\kappa}^l_k(i).\) The same reasoning may be applied to show \(\sum_{p=1}^m\alpha_k^p(i)s^p_k(i)=s^\top_k(i)\alpha_k(i).\) We then have for all \(l\)
\small
\begin{align}
    0\mkern-3mu&=\mkern-5mu\bigg\langle\mkern-5mu x_k^l\mkern-4mu-\mkern-4mu\hat{x}_{k|k-1}^l\mkern-3mu-\mkern-6mu\int_{\mathbb{R}^d}\mkern-4mu\sum_{p=1}^m\mkern-2mu\kappa_k^{l,p}\mkern-2mu(i)s_k^p(i)di,\mkern-5mu\int_{\mathbb{R}^d}\mkern-3mu\sum_{p=1}^m\mkern-2mu\alpha_k^p(i')s^p_k(i')di'\mkern-5mu\bigg\rangle_\mathcal{\mkern-5muH}\nonumber\\
    &=\mkern-5mu\bigg\langle\mkern-5mu x_k^l\mkern-4mu-\mkern-4mu\hat{x}_{k|k-1}^l\mkern-3mu-\mkern-6mu\int_{\mathbb{R}^d}{s}^\top_k(i){\kappa}^l_k(i)di,\int_{\mathbb{R}^d}s^\top_k(i')\alpha_k(i')di'\bigg\rangle_\mathcal{H}\nonumber\\
    &=\mkern-3muE\bigg[\bigg(x_k^l-\hat{x}_{k|k-1}^l-\mkern-4mu\int_{\mathbb{R}^d}\mkern-7mu{s}_k(i)^\top{\kappa}^l_k(i)di\bigg)\mkern-4mu\int_{\mathbb{R}^d}\mkern-7mus^\top_k(i')\alpha_k(i') di'\bigg]\nonumber\\
    &=\mkern-7mu\int_\Theta\mkern-3mu\int_{\mathbb{R}^d}\mkern-7mu\bigg(\mkern-4mux_k^l\mkern-3mu-\mkern-3mu\hat{x}_{k|k-1}^l\mkern-3mu-\mkern-7mu\int_{\mathbb{R}^d}\mkern-9mu{s}_k(i)^{\mkern-4mu\top}{\kappa}^l_k(i)di\mkern-4mu\bigg)s^{\mkern-4mu\top}_k(i')\alpha_k(i') di'dP(\theta).\nonumber
    \end{align}
    \normalsize
 Note that as \(s_k\) is a random field with bounded covariance function, and \(\alpha_k\) is composed of absolutely integrable components, then we may employ Theorem \ref{Fub-Ton-Custom} and interchange expectation and integration operators. Therefore for all \(l,\)
 \small
 \begin{align}
     0\mkern-3mu&=\mkern-8mu\int_{\mathbb{R}^d}\mkern-3mu\int_\Theta\mkern-3mu\bigg(\mkern-3mux_k^l\mkern-3mu-\mkern-3mu\hat{x}_{k|k-1}^l\mkern-3mu-\mkern-4mu\int_{\mathbb{R}^d}\mkern-4mu{s}_k(i)^{\mkern-4mu\top}\mkern-4mu{\kappa}^l_k(i)di\bigg)s^{\top}_k(i')\alpha_k(i') dP(\theta)di'\nonumber\\
     &=\mkern-4mu\int_{\mathbb{R}^d}\mkern-4muE\mkern-2mu\left[\left(\mkern-4mux_k^l-\hat{x}_{k|k-1}^l-\int_{\mathbb{R}^d}{s}_k(i)^\top{\kappa}^l_k(i)di\right)s^\top_k(i')\right]\alpha_k(i')di'\nonumber
 \end{align}
\normalsize
Because \(\alpha_k\) is arbitrary, we must have that for all \(l\) and almost all \(i'\)

\begin{align}
    E\bigg[\bigg(x_k^l-\hat{x}_{k|k-1}^l-\int_{\mathbb{R}^d}{s}^\top_k(i){\kappa}_k^l(i)di\bigg){s}_k^\top(i')\bigg]={0}.\label{eq:constraint}
\end{align}

As the first argument is a scalar and the second a vector, the condition that this holds true for all \(l\) can be conveyed by a single matrix equality. Note that \(s^\top_k(i)\kappa_k^l(i)\) is equivalent to the \(l^{th}\) row of \(\kappa_k(i)s_k(i).\) By stacking the \(l\) row vectors in  \eqref{eq:constraint} to form an \(\mathbb{R}^{n\times m}\) matrix we may impose the matrix equality that for almost all \(i'\)
\begin{align}
    E\bigg[\bigg({x}_k-\hat{x}_{k|k-1}-\int_{\mathbb{R}^d}{\kappa}_k(i){s}_k(i)di\bigg){s}_k^\top(i')\bigg]={0}.\label{eq:matrixequality}
\end{align}

    Substituting in our value for \(s_k(i)\) into \eqref{eq:matrixequality} we find
    \small
    \begin{align}
    &E\bigg[\bigg(x_k-\hat{x}_{k|k-1}-\int_{\mathbb{R}^d}\kappa_k(i)\big[\gamma(i)(x_k-\hat{x}_{k|k-1})+v_k(i)\big]di\bigg)\nonumber\\
    &\times\bigg(\big[\gamma(i')(x_k-\hat{x}_{k|k-1})+v_k(i')\big]\bigg)^\top\bigg]\nonumber\\
    =&E\bigg[(x_k-\hat{x}_{k|k-1})(x_k-\hat{x}_{k|k-1})^\top\bigg]\gamma^\top(i')\nonumber\\
    &-\int_{\mathbb{R}^d}\bigg(\kappa_k(i)(\gamma(i)E\bigg[(x_k-\hat{x}_{k|k-1})(x_k-\hat{x}_{k|k-1})^\top\bigg]\gamma^\top(i')\nonumber\\&+E\bigg[v_k(i)v_k^\top(i')\bigg]\bigg)di=0\nonumber
    \end{align}
    \normalsize
    Let the posterior covariance matrix \(E[(x_k-\hat{x}_k)(x_k-\hat{x}_k)^\top]\) be denoted by \(P_k\), and let the prior covariance matrix \(E[(x_k-\hat{x}_{k|k-1})(x_k-\hat{x}_{k|k-1})^\top]\) be denoted by \(P_{k|k-1}.\)
    Substituting these terms for the corresponding expectations, we find the optimal gain function \(\kappa_k(i)\) must satisfy the condition

    \small
    \begin{align}
    P_{k|k-1}\gamma^\top(i')-\mkern-6mu\int_{\mathbb{R}^d}\mkern-4mu\kappa_k(i)\bigg(\gamma(i)P_{k|k-1}\gamma^\top(i')+R(i-i')\bigg)di=0.\nonumber
    \end{align}
    \normalsize

    Rearranging these terms we arrive at the necessary and sufficient condition 
    \small
    \begin{align}\label{eq:necccondition}
\int_{\mathbb{R}^d}\kappa_k(i)\bigg(\gamma(i)P_{k|k-1}\gamma^\top(i')+R(i-i')\bigg)di=P_{k|k-1}\gamma^\top(i').
\end{align}
\normalsize
\end{proof}

An alternative derivation is given in \cite{VarleyACC} based on the Calculus of Variations approach. The derivation above, which uses the Hilbert Projection Theorem, has the advantage of a guaranteed uniqueness of the correction term \(\int_{\mathbb{R}^d}\kappa_k(i)s_k(i)di\) and state estimate \(\hat{x}_k\), which follows immediately from Theorem \ref{HilbertProjTheorem}. This approach avoids the use of the functional derivative which has attracted some criticism \cite{Omatu}, \cite{Curtain}, as well as closely mirroring the methodology of the original Kalman proof \cite{KalmanOriginal}. 

Recall that we initially assumed \(M_{k-1}\) is a closed subspace. We will now show that an explicit closed-form solution of the optimal gain function can be derived via \(d\)-dimensional Fourier based methods, thanks to the w.s.s. assumption of the measurement noise. We will then show that when a unique gain function which satisfies \eqref{eq:necccondition} exists, it is indeed the case that \(M_{k-1}\) is closed, which is derived in Corollary \ref{corollary:closedM}.

\subsection{An Explicit Solution for the Optimal Gain Function}\label{subsec:ExplicitSol}
We may now proceed to find an explicit value for the optimal gain function that satisfies \eqref{eq:ACCcondition}. Consider the matrix-valued functional
\begin{align}
    F(\kappa_k)\triangleq\bigg(I-\int_{\mathbb{R}^d}\kappa_k(i)\gamma(i)di\bigg)P_{k|k-1}\in \mathbb{R}^{n\times n}\label{eq:functional}
\end{align}
and note that,
\begin{align}
    \int_{\mathbb{R}^d}\kappa_k(i)R(i'-i)di&=F(\kappa_k)\gamma^\top(i')\label{eq:preconv}
\end{align}

The left hand side of  \eqref{eq:preconv} is simply a multi-dimensional convolution. Many useful properties of the scalar Fourier transform carry over to the multi-dimensional Fourier transform. We will make use of the convolution property \cite{Bracewell2003},
\begin{align}
    &\mathcal{F}\left\{\int_{\mathbb{R}^d}f(i_1-y_1,...,i_d-y_d)g(y_1,...,y_d)\right\}dy_1...dy_d\nonumber\\=&\mathcal{F}\{(f\ast g)(i_1,i_2,...,i_d)\}=(\mathcal{F}\{f\})(\omega)(\mathcal{F}\{g\})(\omega).\label{eq:fourierconv}
\end{align}
This enables us to find that \eqref{eq:preconv} is equivalent to
\begin{align}
    (\kappa_k\ast R)(i') &= F(\kappa_k)\gamma^\top(i')\label{eq:postconv}
\end{align}
Note that \(\omega\) will be a vector of the same dimension as \(i\). We now require three assumptions to ensure existence and invertibility of the Fourier transform of \(R\), as well as the existence of the Fourier transform of \(\gamma^\top\):
\assumgamma*
\assumR*
\assumbarR*
Applying \eqref{eq:postconv} and \eqref{eq:fourierconv} to  \eqref{eq:preconv} implies
\begin{align}
     \Bar{\kappa}_k(\omega)\Bar{R}(\omega)&=F(\kappa_k)\Bar{\gamma}^\top(\omega)\nonumber\\
     \Bar{\kappa}_k(\omega)&=F(\kappa_k)\Bar{\gamma}^\top(\omega)\Bar{R}(\omega)^{-1}\nonumber\\
     \kappa_k(i)&=F(\kappa_k)f(i), \ f(i)=\mathcal{F}^{-1}\{\Bar{\gamma}^\top(\omega)\Bar{R}(\omega)^{-1}\}\label{eq:kappadef}.
\end{align}
For this method to be applicable we must also make the following assumption 
\assumRGamma*
Combining \eqref{eq:functional} and \eqref{eq:kappadef} yields
\begin{align}
    F(\kappa_k)&=\bigg(I-F(\kappa_k)\int_{\mathbb{R}^d}f(i)\gamma(i) di\bigg)P_{k|k-1}\nonumber\\
    \implies F(\kappa_k)&=P_{k|k-1}\bigg(I+\int_{\mathbb{R}^d}f(i)\gamma(i)diP_{k|k-1}\bigg)^{-1}\label{eq:functionalsolved}\\
    \implies \kappa_k(i)&=P_{k|k-1}\bigg(I+\int_{\mathbb{R}^d}f(i)\gamma(i)diP_{k|k-1}\bigg)^{-1}f(i)\nonumber
\end{align}
Note that this directly implies that for all \(k\in \mathbb{N}\), \(\kappa_k\in L_1(\mathbb{R}^d,\mathbb{R}^{n\times m})\cap L_2(\mathbb{R}^d,\mathbb{R}^{n\times m})\).
For simplicity of notation, recall from \eqref{def:S} that \(S=\int_{\mathbb{R}^d}f(i)\gamma(i)di\in \mathbb{R}^{n\times n}.\) Having thus found the optimal \(\kappa_k(i)\) function which generates \(\beta_k\), we now modify \eqref{eq:projectionform} to find for any state component \(l\),
\begin{align}
    &\proj_{M_{k-1}+M^-_{k}}x_k^l\nonumber\\
    &=\hat{x}_{k|k-1}^l+\sum_{p=1}^m\int_{\mathbb{R}^d}\big[P_{k|k-1}(I+SP_{k|k-1})^{-1}f(i)\big]^{l,p}s^p_k(i)di.\nonumber
\end{align}
These \(l\) equations may be stacked up into a single matrix equation of the form
\begin{align}
    &\proj_{M_{k-1}+M^-_{k}}x_k\nonumber\\
    &=\hat{x}_{k|k-1}+\int_{\mathbb{R}^d}P_{k|k-1}(I+SP_{k|k-1})^{-1}f(i)s_k(i)di,\label{eq:betterprojectionform}
\end{align}
where the projection operator acting on the state vector is understood to represent the vector of each components projections. We have thus far assumed that \(\hat{x}_{k|k-1}\) is known to us, we will complete the chain now by determining this value explicitly as
\begin{align}
    \hat{x}_{k|k-1}&=\proj_{M_{k-1}}x_k\nonumber\\
    &=\proj_{M_{k-1}}Ax_{k-1}+\proj_{M_{k-1}}w_{k-1}\nonumber\\
    &=A\proj_{M_{k-1}}x_{k-1}\nonumber\\
    &=A\proj_{M_{k-2}+M^+_{k-2}}x_{k-1}\nonumber\\
    &=A\hat{x}_{k-1}\label{eq:stateproj}
\end{align}
The form of \eqref{eq:betterprojectionform} with this substitution is then
\begin{align}
    \hat{x}_k=&\proj_{M_{k-1}+M^-_{k}}x_k\nonumber\\
    =&A{\hat{x}_{k-1}}+P_{k|k-1}(I+SP_{k|k-1})^{-1}\int_{\mathbb{R}^d}f(i)s_k(i)di\nonumber\\
    =&A{\hat{x}_{k-1}}+\nonumber\\&P_{k|k-1}(I+SP_{k|k-1})^{-1}\int_{\mathbb{R}^d}\mkern-5muf(i)\big(z_k(i)-\gamma(i)\hat{x}_{k|k-1}\big)di,\nonumber
\end{align}
where in the last line we substitute \(s_k(i)=z_k(i)-\hat{z}_k(i)\), noting that 
\begin{align}
    \hat{z}_k(i)&=\proj_{M_{k-1}}z_k(i)\nonumber\\
    &=\proj_{M_{k-1}}\big(\gamma(i)x_k+v_k(i)\big)\nonumber\\
    &=\gamma(i)\proj_{M_{k-1}}x_k\nonumber\\
    &=\gamma(i)\hat{x}_{k|k-1}.\nonumber
\end{align}

Finally, we will show that the space \(M_{k-1}\) is indeed closed, which was previously assumed in Section \ref{subsec:OptimalityConditions}. Hence it admits a unique projection by the Hilbert projection theorem (Thm. \ref{HilbertProjTheorem}), implying the uniqueness of the optimal filter derived above.

\begin{corollary}\label{corollary:closedM}
    Given the system defined in Section \ref{sec:ModelDefinition}, and conditions given in Section \ref{subsec:OptimalityConditions}, the space \(M_{k-1}\) is a closed subspace of \(\mathcal{H}\).
\end{corollary}
\begin{proof}
    The optimal gain function given by \eqref{eq:kappadef} represents a kernel defining an operator \(K_k:\mathcal{H}\rightarrow \mathcal{S}_k\) as presented in \eqref{eq:kappaform}. This operator generates a unique \(\beta_k=K_kx_k\) such that for each element \(l\), \((x_k^l-\beta_k^l)\perp \mathcal{S}_k\). This unique \(\beta_k\) is in fact a minimizer as for all \(\tilde{s}_k\in \mathcal{S}_k\)
    \begin{align}
        \|\tilde{s}_k-x_k^l\|_\mathcal{H}^2=&\|\beta_k^l-x_k^l+\tilde{s}_k-\beta_k^l\|_\mathcal{H}^2\nonumber\\
        =&\|\beta_k^l-x_k^l\|_\mathcal{H}^2+\|\tilde{s}_k-\beta_k^l\|_\mathcal{H}^2\nonumber\\
        &+2\langle \beta_k^l-x_k^l, \tilde{s}_k-\beta_k^l\rangle_\mathcal{H}\nonumber\\
        =&\|\beta_k^l-x_k^l\|_\mathcal{H}^2+\|\tilde{s}_k-\beta_k^l\|_\mathcal{H}^2\nonumber\\
        \implies \|x_k^l-\tilde{s}_k\|_\mathcal{H}^2 \geq& \|x_k^l-\beta_k^l\|_\mathcal{H}^2,\nonumber
    \end{align}
    with equality iff \(\tilde{s}_k=\beta_k^l\).
\end{proof}
Let us now examine an arbitrary Cauchy sequence \(\{s_1,s_2,...,s_n\}\) with each element in \(\mathcal{S}_k\). We also denote \(\lim_{n\rightarrow \infty}s_n=\bar{s}\). As \(\mathcal{S}_k\) is a subspace of a closed space \(\mathcal{H}\), it is guaranteed that \(\bar{s}\in\mathcal{H}\) exists. For any \(\bar{s}\in\mathcal{H}\) there exists a unique minimizer \(\beta^l\in \mathcal{S}_k\) such that
\begin{align}
    \|\beta^l-\bar{s}\|_\mathcal{H}^2 &\leq \|s-\bar{s}\|_\mathcal{H}^2 \quad \forall s\in \mathcal{S}_k\nonumber\\
     \|\beta^l-\bar{s}\|_\mathcal{H}^2& \leq \|s_n-\bar{s}\|_\mathcal{H}^2 \quad \forall n.\nonumber
\end{align}
As the upper bound goes to zero as \(n\) increases, and the inequality holds for all \(n\), \(\bar{s}\) must equal \(\beta^l\) and hence also reside in the space \(\mathcal{S}_k\). As the Cauchy sequence was arbitrary, \(\mathcal{S}_k\) contains all its limit points and is hence closed. It is well known that if two closed linear subspaces of a Hilbert space are orthogonal, then the direct sum of these subspaces is also closed \cite[p. 38]{conway}. This result, combined with the fact that \(M_{k}=M_{0}\oplus \mathcal{S}_1\oplus\mathcal{S}_2...\oplus \mathcal{S}_k\), implies that if \(M_0\) is a closed space, \(M_{k}\) is a closed space for any \(k\).

\subsection{The Optimal Filter}\label{subsec:OptimalFilter}
We now have sufficient knowledge to state a key theorem.

\begin{theorem}\label{MainResult}
For the system defined in Section \ref{sec:ModelDefinition}, at each time step \(k\), the optimal filter gain \(\kappa_k(i)\) and associated covariance matrices are given by the equations.
    \begin{align}
        \kappa_k(i)&=P_kf(i),\nonumber\\
        P_{k|k-1}&=AP_{k-1}A^\top+Q,\nonumber\\
        P_k&=P_{k|k-1}\bigg(I+SP_{k|k-1}\bigg)^{-1},\nonumber
    \end{align}
    where \(f(i)=\mathcal{F}^{-1}\{\Bar{\gamma}^\top(\omega)\Bar{R}(\omega)^{-1}\}, \quad S=\int_{\mathbb{R}^d}f(i)\gamma(i)di\).
\end{theorem}
\begin{proof}
The optimal gain function \(\kappa_k\) has been derived in Section \ref{subsec:ExplicitSol} and is given in \eqref{eq:kappadef}. It remains to derive the covariance matrix update equations. First, we connect the prior covariance matrix to the covariance matrix of the previous time step, via
\begin{align}\label{eq:priorcov}
    P_{k|k-1}&=E[(x_k-\hat{x}_{k|k-1})(x_k-\hat{x}_{k|k-1})^\top]\nonumber\\
    &=AE[(x_{k-1}-\hat{x}_{k-1})(x_{k-1}-\hat{x}_{k-1})^\top]A^\top+Q\nonumber\\
    &=AP_{k-1}A^\top+Q.
\end{align}
Then we determine the relation between \(P_k\) and \(P_{k|k-1}\) by first observing that expansion and substitution yields the equation
\small
\begin{align}
    P_k=E[e_ke_k^\top]\mkern-14mu&\nonumber\\
    =P_{k|k-1}&-\int_{\mathbb{R}^d}\mkern-7mu\kappa_k(i)\gamma(i)diP_{k|k-1}-P_{k|k-1}\int_{\mathbb{R}^d}\mkern-7mu\gamma^\top(i)\kappa_k^\top(i)di\nonumber\\&+\int_{\mathbb{R}^d}\kappa_k(i)\gamma(i)diP_{k|k-1}\int_{\mathbb{R}^d}\gamma^\top(i)\kappa_k^\top(i)di\nonumber\\
    &+\int_{\mathbb{R}^d}\int_{\mathbb{R}^d}\kappa_k(i)R(i-i')\kappa_k(i')^\top\,di\,di'.\nonumber
\end{align}
\normalsize
By substitution of \eqref{eq:ACCcondition}, all but the first two terms cancel to give us
\begin{align}\label{covupdateK}
    P_k=\bigg(I-\int_{\mathbb{R}^d}\kappa_k(i)\gamma(i)di\bigg)P_{k|k-1}
\end{align}
But this is precisely the equation for \(F(\kappa_k)\) defined in \eqref{eq:functional}, so from \eqref{eq:functionalsolved} we simply have
\begin{align}\label{eq:cov}
    P_k=F(\kappa_k)=P_{k|k-1}\bigg(I+\int_{\mathbb{R}^d}f(i)\gamma(i)\,diP_{k|k-1}\bigg)^{-1}
\end{align}
\end{proof}

\section{Optimal Linear Filter Stability}\label{sec:Stability}
We will now show that the derived filter provided by Theorem \ref{MainResult} is mean square stable under assumptions A1-A7. To do so we will establish via Lemma \ref{lemma:convergence} that the associated covariance matrices obey a discrete-time algebraic Riccati equation. We will then, in Theorem \ref{theorem:convergence}, show that, so long as the initial covariance matrix is non-negative and symmetric, there exists a unique asymptotic \textit{a priori} covariance matrix that satisfies this Riccati equation.

\begin{lemma}
    \label{lemma:convergence}
    Consider the filter equations provided by Theorem \ref{MainResult} for the system defined in Section \ref{sec:ModelDefinition}.
    The matrix \(S\) is positive semi-definite, it is also finite under assumptions \ref{ass:gamma} and \ref{ass:RGamma}.
    The matrix \(S\) admits a unique square root decomposition \(S=GG\) for some symmetric positive semi-definite \(G\in\mathbb{R}^{n\times n}\).
    Furthermore, for $k \geq 1$, the \emph{a priori} covariance matrices \(P_{k+1 | k}\) satisfy the discrete-time algebraic Riccati equation
    \small
    \begin{align}\notag
        P_{k+1 | k}=& A  P_{k | k-1} A^\top \\\label{eq:riccati} 
        - &A P_{k | k-1}G\left(I+G^\top P_{k | k-1}G\right)^{-1}G^\top P_{k | k-1}A^\top + Q. 
    \end{align}
    \normalsize
\end{lemma}
\begin{proof}
    First we find that \(\bar{R}(\omega)\) is Hermitian as
    \begin{align}
        \bar{R}(\omega)^* =&\int_{\mathbb{R}^d}e^{2\pi j\omega \cdot i}R(i)^\top\,di\nonumber\\
        =&\int_{\mathbb{R}^d}e^{2\pi j\omega \cdot i}R(-i)\,di\nonumber\\
        =&\int_{\mathbb{R}^d}e^{-2\pi j\omega \cdot i}R(i)di=\bar{R}(\omega),\nonumber
    \end{align}
    where \((\cdot)^*\) denotes the complex conjugate transpose. \(\bar{R}(\omega)\) is also positive semi-definite, as for almost all \(\omega\) and all \(u\in \mathbb{C}^m\),
    \begin{align}
        u^*\bar{R}(w)u=&\int_{\mathbb{R}^d}e^{-2\pi ji\cdot \omega}E[u^*v(\ell+i)v(\ell)^\top u]\,di, \ \forall \ell \in \mathbb{R}^d\nonumber\\
        =&\int_{\mathbb{R}^d}e^{-2\pi ji\cdot \omega}E\left[\left(v(\ell)^\top u\right)\left(v(\ell+i)^\top u\right)^*\right]\,di.\nonumber
    \end{align}
    Define the stationary complex-valued random field \(\psi(\ell)\triangleq v(\ell)^\top u \in \mathbb{C}\). Denote the corresponding auto-covariance function and power spectral density as \(K_\psi(i)\) and \(S_\psi(\omega)\) respectively. It then follows that
    \begin{align}
        u^*\bar{R}(\omega)u=&\int_{\mathbb{R}^d}e^{-2\pi j i\cdot\omega}E\left[\psi(\ell)\psi(\ell+i)^*\right]\,di\nonumber\\
        =&\int_{\mathbb{R}^d}e^{-2\pi j i\cdot\omega}K_\psi(i)\,di\nonumber\\
        =& S_\psi(\omega),\nonumber
    \end{align}
    where the last equality is given by the Wiener-Khintchine-Einstein relation \cite[p. 82]{Rasmussen}. As \(S_\psi(\omega)\) is non-negative, \(\bar{R}(\omega)\) is positive semi-definite for almost all \(\omega\). As \(\bar{R}(\omega)\) is invertible by assumption \ref{ass:barR}, \(\bar{R}(\omega)\) is positive-definite for almost all \(\omega\). It follows directly that \(\bar{R}(\omega)^{-1}\) is hermitian and positive-definite for almost all \(\omega\).
    
    It can now be shown that \(S\) is symmetric and positive semi-definite as
    \begin{align}
        S=&\int_{\mathbb{R}^d}f(i)\gamma(i)\,di\nonumber\\
        =&\int_{\mathbb{R}^d}\int_{\mathbb{R}^d}e^{2\pi j i\cdot\omega}\bar{\gamma}(\omega)^\top \bar{R}(\omega)^{-1}\gamma(i)\,d\omega\,di\nonumber\\
        =&\int_{\mathbb{R}^d}\bar{\gamma}(\omega)^\top \bar{R}(\omega)^{-1}\int_{\mathbb{R}^d}e^{2\pi j i\cdot\omega}\gamma(i)\,di\,d\omega\label{eq:swapeq}\\
        =&\int_{\mathbb{R}^d}\bar{\gamma}(\omega)^\top \bar{R}(\omega)^{-1}\mathcal{F}^{-1}\{\gamma(i)\}\,d\omega\nonumber\\
        =&\int_{\mathbb{R}^d}\bar{\gamma}(\omega)^\top \bar{R}(\omega)^{-1}\bar{\gamma}(-\omega)\,d\omega\nonumber\\
        =&\int_{\mathbb{R}^d}\bar{\gamma}(-\omega)^\top \bar{R}(\omega)^{-1}\bar{\gamma}(\omega)\,d\omega.\nonumber
    \end{align}
    
    Due to the fact that \(\bar{R}(\cdot)^{-1}\) is Hermitian and \(\gamma\) is real-valued we then have
    \begin{align}
        S=&\int_{\mathbb{R}^d}\bar{\gamma}(\omega)^* \bar{R}(\omega)^{-1}\bar{\gamma}(\omega)\,d\omega.\nonumber
    \end{align}
     The interchange of integral operators in \eqref{eq:swapeq} must be justified. This is permitted by the Fubini-Tonelli theorem \cite[Cor. 13.9]{schilling_2005} as each element of \(\bar{\gamma}(\omega)^*\bar{R}(\omega)^{-1}\bar{\gamma}(\omega)\) is absolutely integrable. This can be seen by noting that every element of \(\bar{\gamma}(\omega)^*\) is bounded, and that every element of \(\bar{R}^{-1}(\omega)\bar{\gamma}(\omega)=\mathcal{F}\{f(i)^\top\}\) is absolutely integrable, ensuring that every element of \(\bar{\gamma}(\omega)^*\bar{R}(\omega)^{-1}\bar{\gamma}(\omega)\) is a finite sum of products of absolutely integrable functions and bounded functions and hence is itself absolutely integrable.
    
    This ensures \(S\) is symmetric. \(\bar{R}(\omega)^{-1}\) is a hermitian positive-definite matrix and so admits a unique Cholesky decomposition \(\bar{R}(\omega)^{-1}=D^*(\omega)D(\omega)\) with rank\((D(\omega))=m\) for almost all \(\omega\in \mathbb{R}^{d}\). We can then state that for almost all \(\omega\),
    \begin{align}
        S=&\int_\mathbb{R}\left(D(\omega)\bar{\gamma}(\omega)\right)^*D(\omega)\bar{\gamma}(\omega)\,d\omega\nonumber\\
        \implies x^\top Sx =& \|D(\cdot)\bar{\gamma}(\cdot)x\|^2_{L_2}\geq 0, \ \forall x\in \mathbb{R}^n \nonumber
    \end{align}
    hence \(S\) is positive semi-definite.
    As \(S\) is symmetric positive semi-definite, it also admits a unique decomposition \(S=GG \) for some symmetric positive semi-definite \(G=S^{\frac{1}{2}}\in\mathbb{R}^{n\times n}\) \cite[p. 440]{Horn13}, proving the first lemma assertion.

    Now, Procedure \ref{alg:kalman} asserts that
    \begin{align}
        P_{k+1|k}&=A P_{k} A^\top+Q\nonumber\\
        P_{k}&=P_{k|k-1}\left(I+S P_{k|k-1}\right)^{-1}.\label{eq:PtoP-}
    \end{align}
    Substituting the form \(S=GG=GG^\top\) into \eqref{eq:PtoP-} we have that
    \small
    \begin{align}
        P_{k+1|k}=&A P_{k|k-1}\left(I+S P_{k|k-1}\right)^{-1}A^\top+Q\nonumber\\
        =&A P_{k|k-1}\left(I+GG^\top  P_{k|k-1}\right)^{-1}A^\top+Q\nonumber\\
        =&A P_{k|k-1}\left(I-G\left(I+G^\top P_{k|k-1}G\right)^{-1}G^\top P_{k|k-1}\right)A^\top\nonumber\\
        &+Q\nonumber,
    \end{align}
    \normalsize
    which follows from the Woodbury matrix identity. A final expansion allows us to arrive at the canonical form of the discrete-time Riccati equation \eqref{eq:riccati}, completing the proof.
\end{proof}

Lemma \ref{lemma:convergence} is important because it enables us to employ the standard notions of stabilizability and detectability from finite-dimensional linear systems theory to establish the following result characterizing the convergence of our optimal linear filter. Note that Theorem \ref{theorem:convergence} requires the following assumptions
\assumstability*
\assumdetectability*

\begin{theorem}
    \label{theorem:convergence}
    Consider the filter equations given by Theorem \ref{MainResult}, and let assumptions A6 and A7 hold. If \(Q>0\), then the asymptotic \textit{a priori} covariance \(P^-_\infty\) is the unique stabilizing solution of the discrete-time algebraic Riccati equation
    \begin{align}
        P_\infty^-&=A P_\infty^-A^\top-A P_\infty^-G\left(I+G^\top P_\infty^-G\right)^{-1}G^\top P_\infty^-A^\top+Q.\nonumber
    \end{align}
\end{theorem}
\begin{proof}
    Given Lemma \ref{lemma:convergence}, \cite[p. 77]{AndersonandMoore} establishes that discrete-time algebraic Riccati equations of the form \eqref{eq:riccati} converge to a unique limit if \(Q>0\), \((A,G)\) is detectable, and \((A,Q)\) is stabilizable. In other words, if these conditions hold, then for any non-negative symmetric initial covariance \(P_0\), there exists a limit
\begin{align}
    \lim_{k\rightarrow \infty} P_{k|k-1}=P_\infty^-,\nonumber
\end{align}
and this limit satisfies the steady-state equation
\begin{align}
    P_\infty^-=&A P_\infty^-A^\top+Q\nonumber\\
        &-A P_\infty^-G\left(I+G^\top P_\infty^-G\right)^{-1}G^\top P_\infty^-A^\top,\nonumber
\end{align}
completing the proof.
\end{proof}

Given Theorem \ref{theorem:convergence}, the asymptotic \textit{a posteriori} covariance matrix $P_\infty$ can easily be calculated via the covariance update equation of Theorem \ref{MainResult}, namely, $P_\infty = P_\infty^-(I + SP_\infty^-)^{-1}$.
Surprisingly, Theorem \ref{theorem:convergence} requires only existing concepts from finite-dimensional linear system theory, and no concepts from infinite-dimensional systems theory.

\section{Algorithm and Implementation}\label{sec:Implementation}
Theorem \ref{MainResult}, combined with \eqref{eq:stateproj}, can be used to construct an algorithm as follows:
\begin{algorithm}[h]
\caption{Optimal Linear Filter}\label{alg:kalman}
\begin{algorithmic}[1]
\State Inputs: \(A, Q, R(i,i'), P_{0}, \hat{x}_{0}, \gamma(i)\)
\For{\(k\geq 1\)}
\State \(f(i)=\mathcal{F}^{-1}\{\Bar{\gamma}^\top(\omega)\Bar{R}(\omega)^{-1}\}\)
\State \(S=\int_{\mathbb{R}^d}f(i)\gamma(i)di\)
\State \(P_{k|k-1}=AP_{k-1}A^\top+Q\)
\State \(P_k=P_{k|k-1}(I+SP_{k|k-1})^{-1}\)
\State \(\hat{x}_{k|k-1}=A\hat{x}_{k-1}\)
\State Obtain measurement \(z_k(i)\)
\State \(\hat{x}_k=\hat{x}_{k|k-1}+P_k\int_{\mathbb{R}^d}f(i)(z_k(i)-\gamma(i)\hat{x}_{k|k-1})di\)
\EndFor
\State Outputs: \(\{P_1,P_2,...,P_N\}, \{\hat{x}_1,\hat{x}_2,...,\hat{x}_N\}\)
\end{algorithmic}
\end{algorithm}

We will now discuss the computational complexity of Procedure \ref{alg:kalman} and its connection to existing results, before implementing it in a simulated environment in Section \ref{sec:Simulations}.

\subsection{Computational Complexity}\label{subsec:complexity}
Our optimal linear filter has many similar properties to the finite-dimensional Kalman filter, such as the convenient ability to calculate the full trajectory of covariance matrices before run-time, as they do not rely on any measurements. This filter is similar to the filter formulated in previous work \cite{VarleyACC}, although the derivation of the filter presented in this work relies on the methodology of projections rather than the calculus of variations approach. The filter given in this article extends the previously derived filter to operate not only on the index set \(\mathbb{R}\), but on any Euclidian space. Filtering on the Euclidian plane \(\mathbb{R}^2\) is perhaps the most relevant to imaging applications, but there may be numerous applications involving higher dimensions.

In the standard finite-dimensional Kalman filtering context with \(N\) measurements, naive inversion of the matrix \((R+\Gamma P\Gamma^\top)\in \mathbb{R}^{N\times N}\) requires nearly cubic complexity in \(N\) \cite{Strassen}. This condition can sometimes be ameliorated, however, by observing that if \(R\) is the covariance matrix of a stationary process then it will have a Toeplitz form. This Toeplitz structure allows greatly reduced computational complexity for inversion operations, see \cite{FERREIRA} and the references throughout for a brief summary on the types of Toeplitz matrices and corresponding inversion complexities. For positive-definite Toeplitz matrices in particular, inversion may be performed with \(O(Nlog^2(N))\) complexity.

The optimal Kalman filter in Procedure \ref{alg:kalman} possesses similar advantages due to the stationary structure of \(R\). If \(N\) measurement samples are taken, we see in Procedure \ref{alg:kalman} that Fourier operations are used to calculate the covariance matrices and optimal gain functions, resulting in a fast Fourier transform applied to an \(m\times n\) matrix with complexity \(O\left(mnN\log(N)\right)\) \cite[p. 41]{FFTs}.

Although Procedure \ref{alg:kalman} does contain matrix inversion, this is confined to low-dimensional \(n\times n\) matrices. This approach to analysis in the continuum provides valuable insights into the optimal gain functions and the behavior of the system. This extension to the continuous measurement domain does introduce a trade-off between accuracy and speed of computation, as any implementable estimator will not be able to sample continuously over the entire space. A standard implementation may sample the space uniformly, which is the approach we take in Section \ref{sec:Simulations}, but a more sophisticated approach could involve the analysis of different non-uniform schemes that exploit the underlying system structure. Examination of this trade-off for various systems is an important area of future work.
\subsection{Comparisons to Existing Results}\label{subsec:comparisons}
During our analysis, we found the necessary and sufficient conditions on the optimal gain satisfies equation \eqref{eq:necccondition}. If the measurements were finite-dimensional, this equation reduces to a matrix equation. The gain may then be calculated via matrix inversion as
\begin{align}
    K_k=P^-_k\gamma^\top\left(R+\gamma P_k^-\gamma^\top \right)^{-1}.\nonumber
\end{align}
Note that this is the standard equation for the Kalman filter in the finite-dimensional case.

Previous works in this area have also derived condition \eqref{eq:necccondition} \cite{Meditch1971},  \cite{Tzaf1973}, \cite{Omatu}. These works, however, have either only defined the optimal gain implicitly, or defined it via an operator inverse \((\cdot)^\dagger\) which satisfies
\begin{align}
    \int_\mathcal{D}G(i,i')G^\dagger(i',i_1)di'=I\delta(i-i_1),\nonumber
\end{align}
where \(\delta(\cdot)\) is the Dirac delta function \cite[Eq. (14)]{Omatu}. The optimal gain is then expressed using this distributed-parameter matrix inverse as
\begin{align}
    K_k=P_k^-\gamma^\top (R+\gamma P_k^-\gamma^\top )^\dagger.\nonumber
\end{align}
Some derivations also give results where the gain is reliant on \(R^\dagger.\) These expressions make the connection to the finite-dimensional Kalman filter clearer, but unfortunately the operator inverse is not guaranteed to exist. Furthermore, some forms of \(R(i,i')\) will lead to inverses that are troublesome to implement. The derivation of the inverse corresponding to an Ornstein-Uhlenbeck process has been presented previously in \cite{VarleyACC} and results in an \(R^\dagger\) defined in terms of the Dirac delta function and its derivatives. The Dirac delta function is well-defined only when under an integral, and requires special care when implemented on real-world hardware. Derivatives of the Dirac delta function will also cause difficulties in implementation. The filter equations derived in this work do not rely on such an inverse, and give an explicit definition of the optimal gain. The formulation presented will, in some cases, avoid the need for generalized functions such as the aforementioned Dirac delta function and its higher derivatives.

The derived prior covariance update equation \eqref{eq:riccati} can be compared with standard forms for the finite-dimensional Kalman filter \cite[p. 39, Eq. (1.9)]{AndersonandMoore}. It is interesting to note that this implies that the prior covariance matrix for our filter is equivalent to that of a finite-dimensional Kalman filter with the measurement noise being a standard multi-variate normal distribution and an observation matrix equal to \(G=S^{\frac{1}{2}}\).

We also highlight that the state estimation error previously defined as \(e_{k} \triangleq x_k-\hat{x}_k\) satisfies
    \begin{align*}
        e_{k+1}=M_ke_k+n_k,
    \end{align*}
    where \(M_k=A-K_k\gamma \in \mathbb{R}^{n\times n}\), and \(n_k=- K_k v_k+w_k\in \mathbb{R}^n.\)
 Hence, the state estimation error dynamics and the optimal (component) gains $\kappa_k(i)$ to apply to points $i \in \mathbb{R}^d$ in the measurement domain depend on the state dynamics of the system, $A$.
 In particular, as $k\to\infty$, the steady-state gain operator $K_k$ must be chosen such that $A-K_k\gamma$ is stable and yields bounded mean square estimation errors. In other words, the gain must take account of the underlying dynamics. This aspect of the filter is important as many standard methods of processing high-dimensional measurements, such as vision-based systems, perform the detection and weighting of measurement features separately from the system dynamics.

\section{Simulation Results} \label{sec:Simulations}

This section presents simulation results for the derived Kalman filter with infinite-dimensional measurements, as described in Algorithm \ref{alg:kalman}. This system consists of a state vector \(x_k=[q_k,\dot{q}_k]^\top \), which represents the position and velocity of some agent. The motion model consists of a state transition model where \(\Delta\) is the change in time between each discrete step. A two-dimensional image is fixed perpendicular to the line of motion of the filter. This image is represented by the function \(C(p):\mathbb{R}^2\rightarrow \mathbb{R}\). We note that the measurement readings are scalar valued, perhaps representing visual light intensity, but this could easily be extended to multi-dimensional image outputs. For example, the intensity of three distinct RGB channels could be represented solely by changing the form of the function \(C(p)\) to be vector-valued. The measurement model used is the simple pinhole camera model. An illustration of this model is given in Fig. \ref{fig:pinhole}. Both the motion model and the measurement model are affected by zero-mean, stationary, additive Gaussian noise such that \(E[w_kw_k^\top ]=Q\) and \(E[v_k(i)v_k(i')^\top ]=R(i-i').\)

\begin{figure}[b]
    \centering
    \includegraphics[scale=0.855]{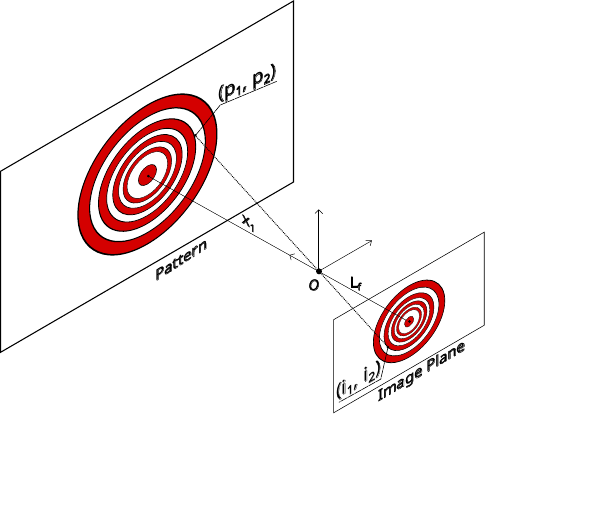}
    \caption{Diagram of the 2-D Pinhole Camera Model}
    \label{fig:pinhole}
\end{figure}

The patterned wall possesses a gray-scale intensity function
\begin{align}
    C(p)=e^{-(\eta \|p\|_2)^2}\cos(\xi \|p\|_2)+1\nonumber
\end{align}

\begin{figure*}[t]\label{fig:observations}
\centering
\begin{subfigure}{0.5\textwidth}
    \includegraphics[scale=0.162]{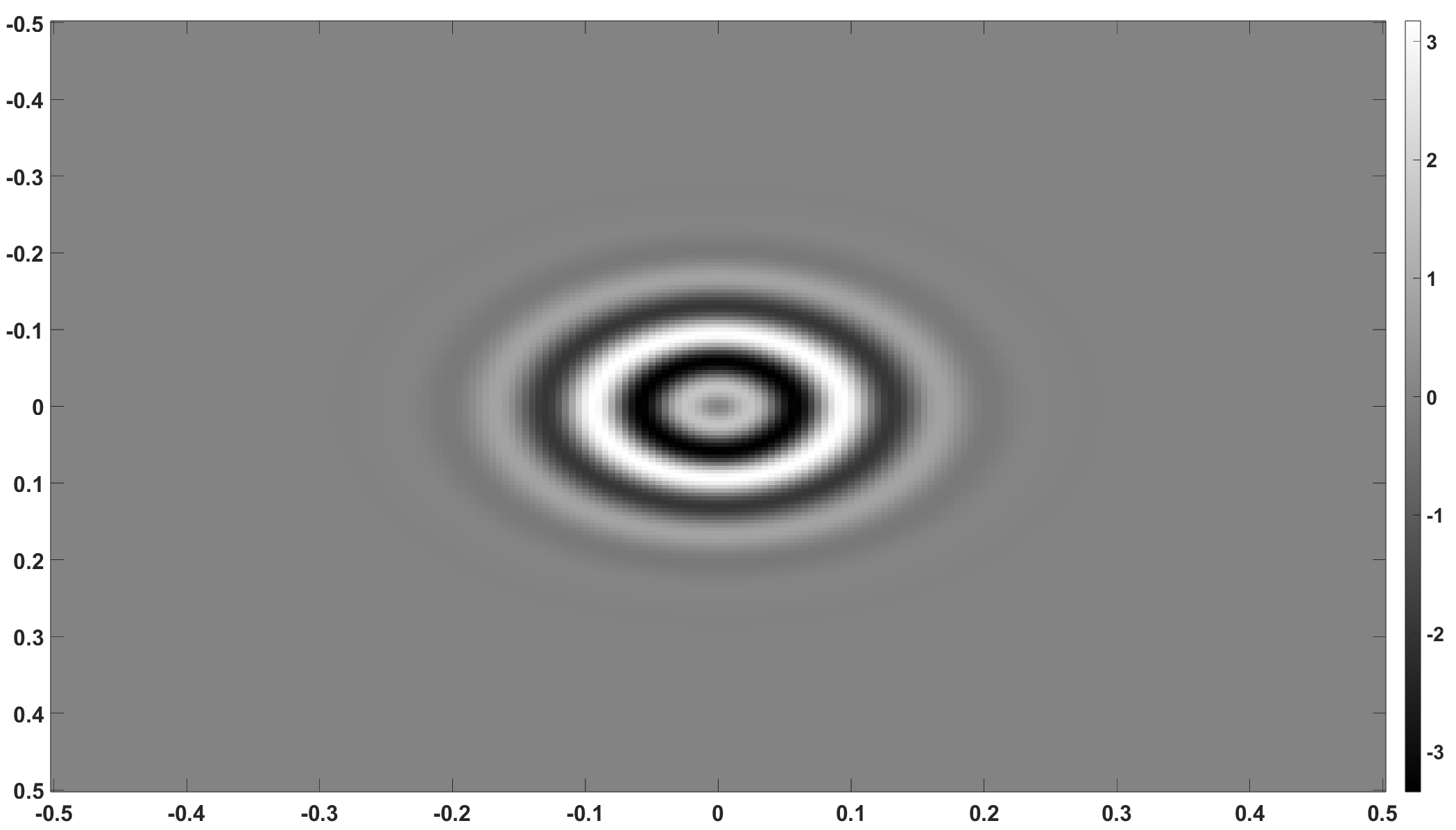}
    \caption{A measurement with no noise.}
    \label{fig:obsnonoise}
    \end{subfigure}\begin{subfigure}{0.5\textwidth}
    \includegraphics[scale=0.162]{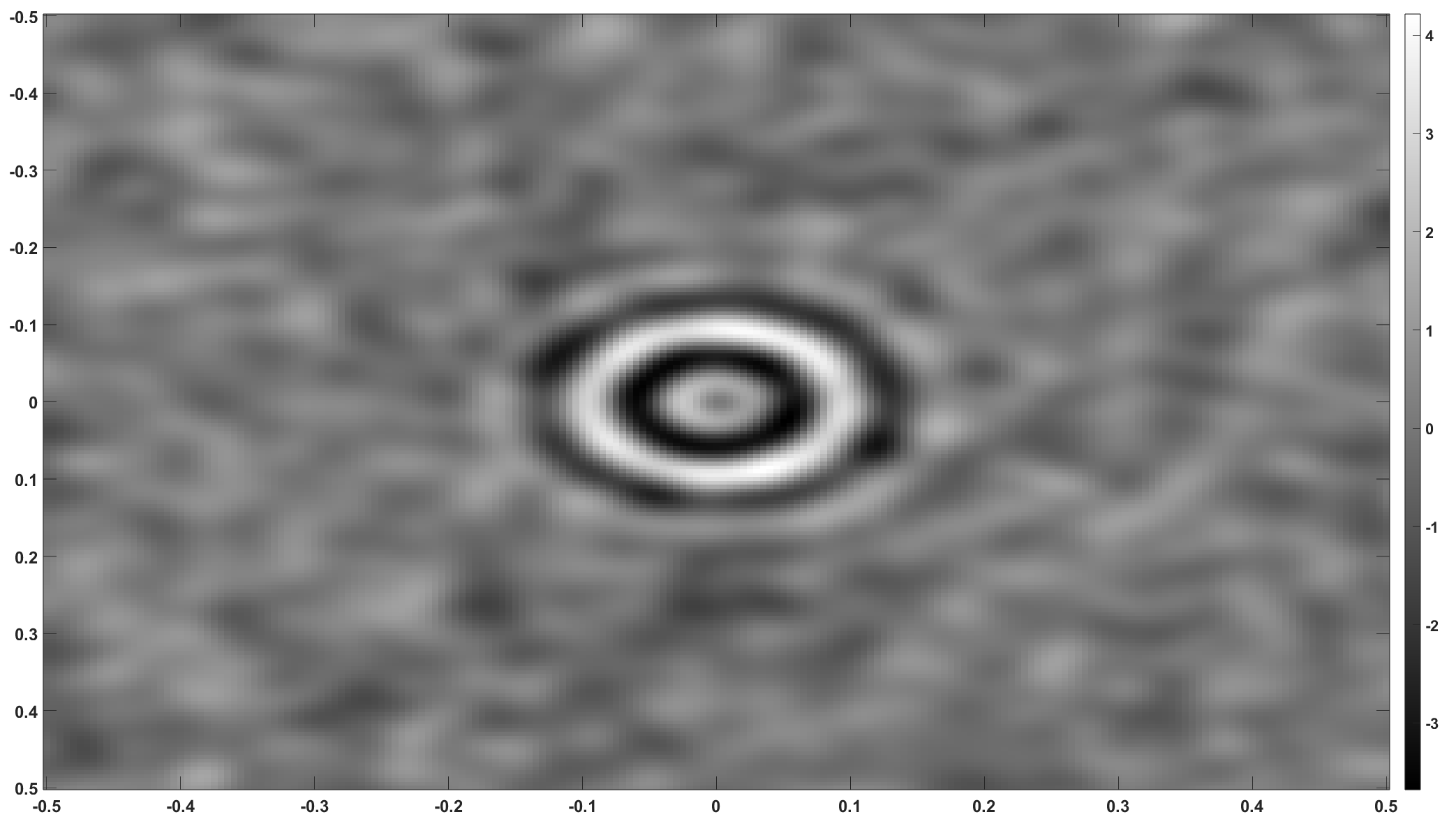}
    \caption{A measurement with substantial noise.}
    \label{fig:obsnoise}
    \end{subfigure}
    \caption{Two measurements of the decaying sinusoidal pattern, Fig. \ref{fig:obsnonoise} has no additive noise, while Fig. \ref{fig:obsnoise} is perturbed by the addition of a stochastic field with a squared exponential covariance function. Note that in \ref{fig:obsnoise} the measurements are spatially discretized for computational reasons.}
\end{figure*}

\begin{small}
\begin{table}[t!]
\begin{center}
\caption{Simulation Parameters}
\label{table:2}
\begin{tabular}{c c c} 
 \hline
 System Variable & Notation & Value\\ [0.5ex] 
 \hline
 State Matrix & \(A\) & \(\begin{bmatrix}1& \Delta\\0&1 \end{bmatrix}\) \\
 Process Noise Covariance & \(Q\) & \(\begin{bmatrix}\sigma_q^2 & 0 \\ 0 & \sigma_{\dot{q}}^2 \end{bmatrix}\)\\ 
 Measurement Kernel & \(\gamma(i)\) & \(\text{See }\eqref{eq:gamma}\)  \\
 State Vector & \(x_k\) & \([q_k,\dot{q}_k]^\top \) \\
 Initial State & \(x_0\) & \([1,0]^\top \)\\
 Linearization Point & \(\Bar{x}\) & \([1,0]^\top\)\\ 
 Integral Domain & \(\mathcal{D}\) & \([-0.5,0.5]^\top \times [-0.5,0.5]^\top\)  \\
 Wall Pattern & \(C(p)\) & \(e^{-(\eta\|p\|_2)^2}\cos(\xi\|p\|_2)+1\)\\
 Measurement Covariance & \(R(i,i')\) & \(\frac{\nu}{2\pi \ell^2}e^{-\|i-i'\|_2^2(2\ell^2)^{-1}}\)\\
 Initial Error Covariance & \(P_0\) & \(Q\)\\
 Initial State Estimate & \(\hat{x}_0\) & \(x_0\)\\
 Time Interval & \(\Delta\) & \(1\)  \\
 Position Variance & \(\sigma_q^2\) & \(0.01\)\\
 Velocity Variance & \(\sigma_{\dot{q}}^2\) & \(0.01\)\\
 Decay Parameter & \(\eta\) & \(0.1\) \\ 
 Frequency Parameter & \(\xi\) & \(0.8\) \\ \
 Measurement Noise Intensity & \(\nu\) & \(10\) \\ 
 Focal Length & \(L_f\) & \(0.01\) \\
 Length Scale & \(\ell\) & \(0.025\) \\
 Sample Spacing & \(\Delta_s\) & \(0.005\)\\ [1ex]
 \hline
\end{tabular}
\end{center}
\end{table}
\end{small}

where \(\eta\) is a scalar parameter that determines the decay rate and \(\xi\) is a scalar parameter that determines the frequency.
Utilizing the pinhole camera model, our non-linear measurement equation is then of the form
\begin{align}
    z_k(i,q_k)&=C\left(\frac{iq_k}{L_f}\right)+v_k(i).\nonumber
\end{align}
This function is linearized around some equilibrium point \(\Bar{x}=[\Bar{q},\Bar{\dot{q}}]^\top\) to derive the linear form used in the simulation. This linearization is given in equation \eqref{eq:obs_linear} and leads to the corresponding measurement function

\small
\begin{align}
    \gamma(i) = &\bigg[-e^{-(\eta L_f^{-1}\Bar{q} \|i\|_2)^2}
    \times \big(2(\eta L_f^{-1}\|i\|_2)^2\Bar{q}\cos(\xi L_f^{-1}  \|i\|_2\Bar{q})\nonumber\\&+\xi L_f^{-1}\|i\|_2\sin(\xi L_f^{-1}\|i\|_2\Bar{q})\big),0\bigg]\in\mathbb{R}^{1\times 2}.\label{eq:gamma}
\end{align}
\normalsize

The system equations then are:
\begin{align}
    x_{k+1}=&\begin{bmatrix} 1&\Delta\\0&1\end{bmatrix}\begin{bmatrix} q_k\\\dot{q}_k\end{bmatrix}+w_k\nonumber\\
    z_k(i)=&-e^{-(\eta L_f^{-1}\Bar{q} \|i\|_2)^2}\big[2(\eta L_f^{-1}\|i\|_2)^2\Bar{q} \cos(\xi  L_f^{-1}\|i\|_2\Bar{q})\nonumber\\&+\xi L_f^{-1}\|i\|_2\sin(\xi L_f^{-1}\|i\|_2\bar{q})\big]q_k+v_k(i). \label{eq:obs_linear}
\end{align}
\normalsize

An example of the linearized measurement pattern as measured by the system without and with noise are presented in Figs. \ref{fig:obsnonoise} and \ref{fig:obsnoise} respectively.

We can formulate a discrete-time Riccati equation for this system as given by \eqref{eq:riccati}.

The matrix pair \((A,Q)\) is stabilizable in the sense of Definition \ref{def:stabilizable} because considering a matrix \(M=\frac{1}{2\sigma_q^2}I\), it is then clear that 
\begin{align}
    \rho(A-MQ)&=\rho\left(A-\frac{1}{2}I\right)\nonumber\\
    &=\rho\left(\begin{bmatrix}
        \frac{1}{2} & 1 \\ 0 & \frac{1}{2}
    \end{bmatrix}\right) =\frac{1}{2}.\nonumber
\end{align}
It is also true that the matrix pair \((A, G)\) is detectable in the sense of Definition \ref{def:detectable}, where \(G=S^{\frac{1}{2}}\). This can be shown by demonstrating that the matrix pair \((A^\top, G)\) is stabilizable, specifically, in the simulated system the matrix \(G=S^{\frac{1}{2}}\) takes the form
\begin{align}
    S=\begin{bmatrix}
        G_1 & 0 \\ 0 & 0 
    \end{bmatrix}, \ G_1\in \mathbb{R}^+.\nonumber
\end{align}
Considering a matrix \(M\) of the form
\begin{align}
    M=\begin{bmatrix}
        \frac{1}{G_1} & \frac{1}{4G_1} \\ 0 & 0 
    \end{bmatrix}\nonumber
\end{align}
immediately gives that\
\begin{align}
    \rho(A^\top -SM)&=\rho\left(\begin{bmatrix}
        0 & -\frac{1}{4} \\ 1 & 1 
    \end{bmatrix}\right) =\frac{1}{2}.\nonumber
\end{align}

\begin{figure*}[t]\label{fig:pos_vel_traj}
\centering
\begin{subfigure}{0.5\textwidth}
    \includegraphics[scale=0.35]{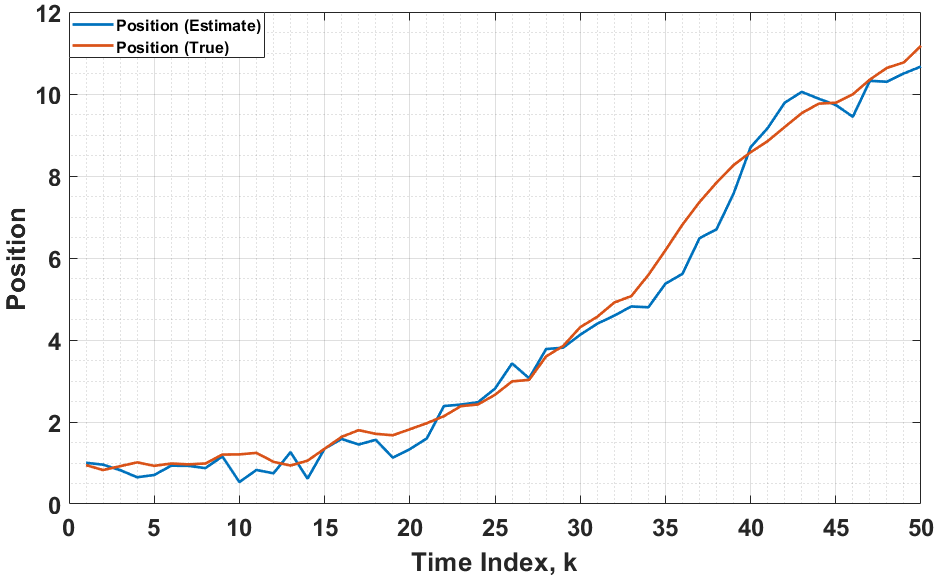}
    \caption{True Position and Estimated Position.}
    \label{fig:Position_Trajectory}
    \end{subfigure}\begin{subfigure}{0.5\textwidth}
    \includegraphics[scale=0.156]{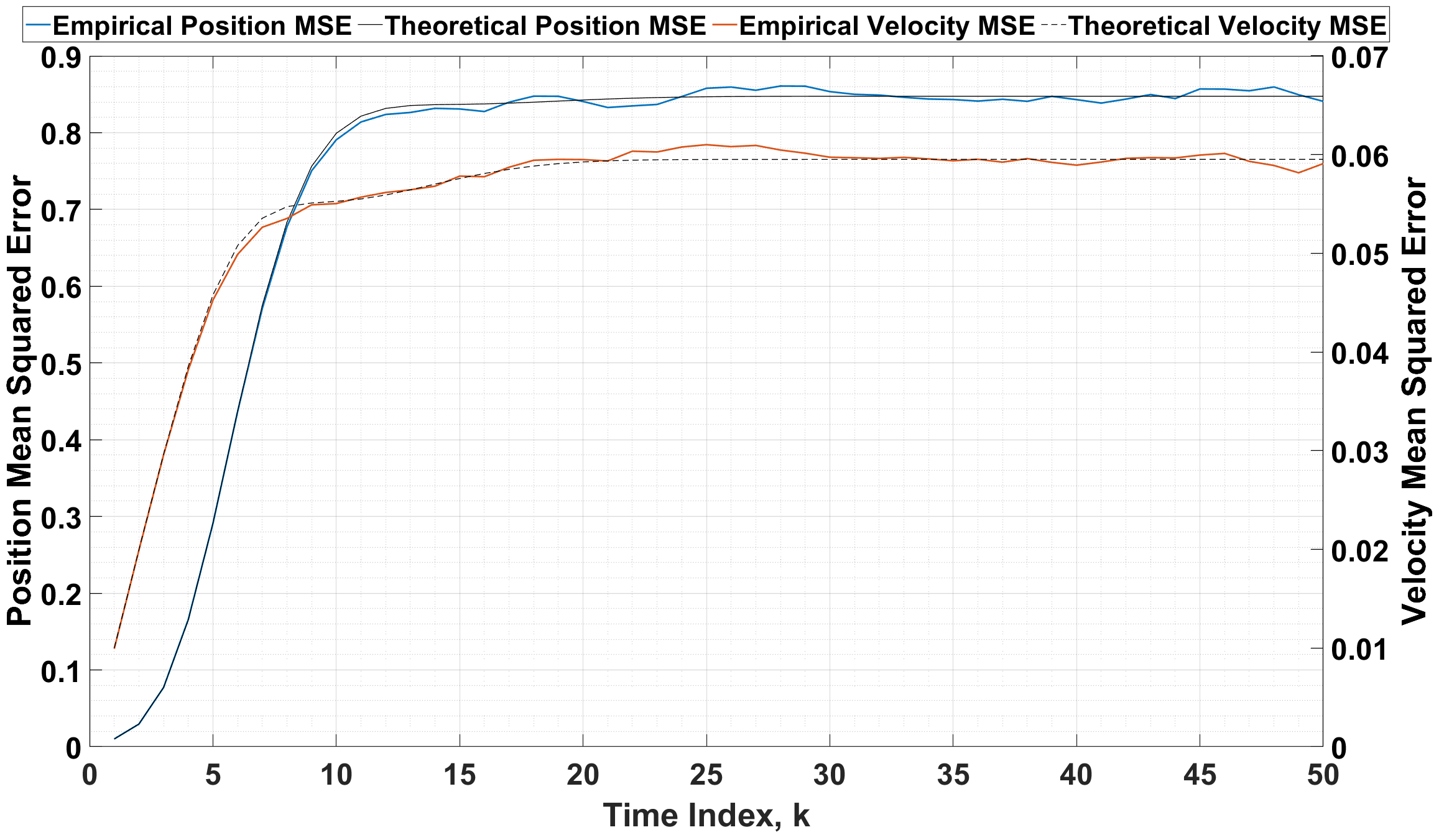}
    \caption{Empirical and Theoretical MSE.}
    \label{fig:MSE_Trajectories}
    \end{subfigure}
    \caption{Fig. \ref{fig:Position_Trajectory} displays the true position and estimated position for a single realization of the filter over 50 time units. Fig. \ref{fig:MSE_Trajectories} displays the empirical position and velocity MSE, averaged over \(20,000\) trials and compared with the theoretical position and velocity MSE.}
\end{figure*}

These conditions, along with the fact that \(Q\) is positive-definite, ensure that the discrete-time Riccati equation for this system possesses a unique stabilizing solution \cite{AndersonandMoore}.
Utilizing the Matlab function \texttt{idare} and substituting the appropriate terms, the following covariance matrices were calculated for this system.

    \begin{align}\label{covariances}
        P_\infty^-=\begin{bmatrix}
        1.2018 & 0.2019\\ 0.2019 & 0.0695
    \end{bmatrix}, \ P_\infty=\begin{bmatrix}
        0.8475 & 0.1424\\ 0.1424 & 0.0595
    \end{bmatrix}.
    \end{align}

The process noise of the position and velocity are uncorrelated. A centred, stationary, stochastic field is used to represent the measurement noise, the covariance function of which is a squared exponential also given in table \ref{table:2}.

With the parameters given in Table \ref{table:2}, Algorithm \ref{alg:kalman} was implemented and simulated using Matlab. It is necessary during implementation to discretise the system appropriately for numerical computation. Numerical functions were evaluated with spacings equal to \(\Delta_s\) in each dimension, or a frequency of \(200\) samples per unit of \(i\) in each dimension. Algorithm \ref{alg:kalman} requires numerical integration for the posterior state update step, and the \texttt{trapz} function was used to implement this numerical integration. For this numerical integration step, a finite pixel domain $\mathcal{D} =[-0.5,0.5]\times [-0.5,0.5]$ was used. This is larger than the effective pixel-width of \(\gamma\) and much larger than the effective width $\ell$ of the noise covariance $R$, suggesting minimal impact on the results obtained. 

This naive method of numerical integration does not exploit the structure of \(f(i)\) and hence there may be more efficient numerical methods that do exploit this structure and lead to more efficient computation. Analysis of such techniques and their impact on efficiency will continue to be an important question for future work in this area.

A single realization of the systems true state and the filter estimate over \(50\) time steps is given in Fig. \ref{fig:Position_Trajectory}. The system was subject to \(20,000\) simulated trials and the average MSE of both the position and velocity trajectories calculated. Fig. \ref{fig:MSE_Trajectories} presents these results alongside the theoretical mean square errors according to \eqref{covariances}.

\newpage
\section{Conclusion}
This article has presented a Kalman filter for a linear system with finite-dimensional states and infinite-dimensional measurements, both corrupted by additive random noise. The assumption that the measurement noise covariance function is stationary allows a closed-form expression of the optimal gain function, whereas previous derivations rely on an implicitly defined inverse. Conditions are derived which ensure the filter is asymptotically stable and, surprisingly, these conditions apply to finite-dimensional components of the system. An algorithm is presented which takes advantage of this stationarity property to reduce computational complexity, and the derived filter is tested in a simulated environment motivated by a linearization of the pinhole camera model. 

An extension to a continuous-time framework is possible, but would require additional tools and some further technical conditions, such as the use of stochastic integrals (either in the It\^o or Stratonovich sense). An extension of the Fubini-Tonelli theorem given by Theorem \ref{Fub-Ton-Custom} would also need to be modified to validate the interchange of expectation and integration on stochastic fields over a continuous spatial and temporal field.

Future work includes investigating more efficient means of executing this algorithm, such as non-uniform integration intervals, to yield insight into effective sensor placement. This could be thought of as a form of feature selection, and this principled approach to feature selection may provide insights related to existing heuristics-based and data driven feature selection strategies often employed in high-dimensional measurement processing. It should also be noted that most real-world applications possess non-linear measurement equations and adapting the linear filter to nonlinear systems will be an important area for future work.

\appendices
\section{Existence of Integral \eqref{eq:kappaform} via Fubini-Tonelli Theorem}\label{app:Fub-Ton}
We introduce one final technical assumption:

\begin{assumption}\label{ass:measurable} 
    \(\kappa_k(\cdot)s_k(\cdot,\cdot)\) is measurable with respect to the product \(\sigma\)-algebra on \(\mathbb{R}^d\times \Theta\).
    \end{assumption}    
\begin{proof}
    For any fixed \(i\), Jensen's inequality tells us 
    \begin{align}
    E\left[\|\kappa(i)s(i)\|_2\right]^2\leq &E\left[\|\kappa(i)s(i)\|_2^2\right]\nonumber\\
        = &E\bigg[Tr\big[\kappa(i)s(i)s^\top(i)\kappa^\top(i)\big]\bigg]\nonumber\\
        = &Tr\bigg[\kappa(i)E\left[s(i)s^\top(i)\right]\kappa^\top(i)\bigg]\nonumber\\
        = &Tr\bigg[\kappa(i)\left(\Sigma(0)+\gamma(i)P\gamma(i)^\top\right)\kappa^\top(i)\bigg] \nonumber\\
        =&\|\Sigma^{\frac{1}{2}}(0)\kappa(i)\|^2_F+\|P^{\frac{1}{2}}\gamma(i)\kappa(i)]\|^2_F\nonumber\\
        \leq& \|\Sigma^{\frac{1}{2}}(0)\|_F^2\|\kappa(i)\|_F^2+\|P^{\frac{1}{2}}\|_F^2\|\gamma(i)\kappa(i)\|_F^2.\nonumber
    \end{align}
    This implies, due to the non-negativity of norms, that 
    \begin{align}
        E\left[\|\kappa(i)s(i)\|_2\right] \leq& \|\Sigma^{\frac{1}{2}}(0)\|_F\|\kappa(i)\|_F+\|P^{\frac{1}{2}}\|_F\|\gamma(i)\kappa(i)\|_F.\nonumber
    \end{align}
    As this is true for any given \(i\), we have the inequality
    \begin{align}
     \int_{\mathbb{R}^d}E\left[\|\kappa(i)s(i)\|_2\right]di \leq&\|\Sigma^{\frac{1}{2}}(0)\|_F\int_{\mathbb{R}^d}\|\kappa(i)\|_F\,di\nonumber\\
     &+\|P^{\frac{1}{2}}\|_F\int_{\mathbb{R}^d}\|\gamma(i)\kappa(i)\|_F \,di.\nonumber
    \end{align}
    As \(\Sigma\) is bounded and \(\kappa(i)\) is absolutely integrable, the first term of the right hand side is finite. By H\"older's inequality and assumptions \ref{ass:gamma} and \ref{ass:kappa} we have \(\|\kappa \gamma\|_{L_1}\leq \|\kappa\|_{L_1}\|\gamma\|_{L_\infty}<\infty\) and hence the second term as well as the the left hand side is finite. As all norms on \(\mathbb{R}^n\) are equivalent, we have
    \begin{align}
    \int_{\mathbb{R}^d}E\left[\|\kappa(i)s(i)\|_1\right]di& < \infty \nonumber\\
     \iff\int_{\mathbb{R}^d}\int_\Theta\|\kappa(i)s(i,\theta)\|_1dP(d\theta)di&<\infty\nonumber\\
    \iff \int_{\mathbb{R}^d}\int_\Theta\left|[\kappa(i)s(i,\theta)]^{\ell}\right|dP(d\theta)di&<\infty,\nonumber
    \end{align}
    where \([\kappa(i)s(i,\theta)]^{\ell}\) is the \(\ell^{th}\) element of the vector \(\kappa(i)s(i,\theta)\).
    As the integral of the absolute expectation of each component is finite, we may apply the Fubini-Tonelli theorem component-wise for each \(\ell\) and compose the vector form
    \begin{align}    \int_{\mathbb{R}^d}E\left[\kappa(i)s(i)\right]di=E\left[\int_{\mathbb{R}^d}\kappa(i)s(i)di\right].\nonumber
    \end{align}
\end{proof}

\bibliographystyle{ieeetr}
\bibliography{references}

\begin{IEEEbiography}[{\includegraphics[width=1in,height=1.25in,clip,keepaspectratio]{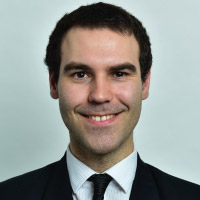}}]{Maxwell M. Varley}
Received the B.S. degree in electrical systems and the M.S. degree in Electrical Engineering from the University of Melbourne, Australia, in 2017 and 2019, respectively. He is currently working towards the Ph.D. degree in electrical engineering at the University of Melbourne.
\end{IEEEbiography}

\begin{IEEEbiography}[{\includegraphics[width=1in,height=1.25in,clip,keepaspectratio]{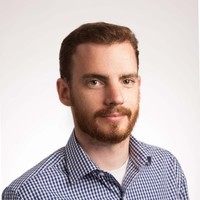}}]{Timothy L. Molloy} (Member, IEEE) was born in Emerald, Australia. He received the B.E. and Ph.D. degrees from Queensland University of Technology (QUT), Brisbane, QLD, Australia, in 2010 and 2015, respectively. He is currently a Senior Lecturer in Mechatronics at the Australian National University. From 2017 to 2019, he was an Advance Queensland Research Fellow at QUT, and from 2020 to 2022, he was a Research Fellow at the University of Melbourne. His interests include signal processing and information theory for robotics and control.
Dr. Molloy is the recipient of a QUT University Medal, a QUT Outstanding Doctoral Thesis Award, a 2018 Boeing Wirraway Team Award, and an Advance Queensland Early-Career Fellowship.
\end{IEEEbiography}

\begin{IEEEbiography}[{\includegraphics[width=1in,height=1.25in,clip,keepaspectratio]{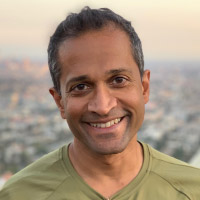}}]{Girish N. Nair} (Fellow, IEEE) was born in Malaysia. He received the B.E. (First Class Hons.) degree in electrical engineering, the B.Sc. degree in mathematics, and the Ph.D. degree in electrical engineering from the University of Melbourne, Parkville, VIC, Australia, in 1994, 1995, and 2000, respectively. He is currently a Professor with the Department of Electrical and Electronic Engineering, University of Melbourne. From 2015 to 2019, he was an ARC Future Fellow, and from 2019 to 2024, he is the Principal Australian Investigator of an AUSMURI project. He delivered a semiplenary lecture at the 60th IEEE Conference on Decision and Control, in 2021. His research interests are in information theory and control. 
Prof. Nair was a recipient of several prizes, including the IEEE CSS Axelby Outstanding Paper Award in 2014 and a SIAM Outstanding Paper Prize in 2006.
\end{IEEEbiography}

\end{document}